\documentclass{article}
\usepackage{graphicx} 
\usepackage[utf8]{inputenc}
\usepackage{amssymb, amsmath, amsthm}
\usepackage{tikz}
\usepackage{bm}
\usepackage{xspace}
\usepackage[bb=boondox]{mathalfa}
\usetikzlibrary{arrows.meta}
\allowdisplaybreaks

\usetikzlibrary{arrows}
\usepackage{pgf}
\pgfarrowsdeclarecombine*{spaced stealth'}{spaced stealth'}{stealth'}{stealth'}{space}{space}
\tikzset{>=spaced stealth'}

\usepackage{xcolor}
\usepackage{tcolorbox}
\usepackage{fullpage}

\usepackage[colorinlistoftodos]{todonotes}
\newcommand{\notelr}[1]{}
\newcommand{\notesl}[1]{}

\newcommand{\dimx}{{d}}

\newcommand{\nconstraints}{{m}}
\newcommand{\listslength}{{\Phi}}
\newcommand{\sizebound}{{\Delta}}

\newcommand{\size}{\mathrm{size}}

\usepackage{algorithm}
\usepackage{algpseudocode}

\algtext*{EndWhile}{}
\algtext*{EndFor}{}
\algtext*{EndIf}{}

\newtheorem{theorem}{Theorem}[section]

\newtheorem{claim}[theorem]{Claim}

\newtheorem{definition}[theorem]{Definition}
\newtheorem{lemma}[theorem]{Lemma}

\usepackage{enumitem}
\usepackage{multicol}

\renewcommand{\epsilon}{\varepsilon}

\newcommand{\R}{\mathbb{R}}
\newcommand{\RR}{\mathbb{R}_{\ge 0}}
\newcommand{\Z}{\mathbb{Z}}
\newcommand{\ZZ}{\mathbb{Z}_{\ge 0}}
\newcommand{\N}{\mathbb{N}}

\newcommand{\capa}{{\mathrm{cap}}}

\newcommand{\rmT}{{\mathrm{T}}}

\newcommand{\ceil}[1]{\left\lceil #1\right\rceil}
\newcommand{\floor}[1]{\left\lfloor #1\right\rfloor}
\newcommand{\poly}{\mathrm{poly}}
\newcommand{\opt}{\mathrm{opt}}
    
\newcommand{\E}{\mathbb{E}}

\mathchardef\mhyphen="2D

\newcommand{\base}{\mathrm{base}}

\newcommand{\bfl}{\mathbf{l}}

\newcommand{\conv}{\mathtt{conv}}
\newcommand{\cost}{\mathrm{cost}}
\newcommand{\pack}{\mathrm{pack}}
\newcommand{\ellv}{\bm{\ell}}

\newcommand{\OPT}{\mathrm{OPT}}

\newcommand{\LP}{\mathbf{\mathrm{LP}}}
 \newcommand*{\T}{^{\mkern-1.5mu\mathsf{T}}}
\newcommand{\lars}[1]{\textcolor{red}{Lars: #1}}

\newcommand{\bun}[1]{\textcolor{teal}{Bun: #1}}



\title{Randomized Rounding over Dynamic Programs}
\author{\'Etienne Bamas\thanks{EPFL. Part of this work was done while the author was a Post-Doctoral Fellow at the ETH AI Center. Email: \texttt{etienne.bamas@epfl.ch}} \and Shi Li\thanks{School of Computer Science, Nanjing University. The work is supported by the State Key Laboratory for Novel Software Technology, and the New Cornerstone Science Foundation. Email: \texttt{shili@nju.edu.cn}} \and Lars Rohwedder\thanks{University of Southern Denmark. Email: \texttt{rohwedder@sdu.dk}}}
\date{}

\begin{document}

\maketitle

\begin{abstract}
We show that under mild assumptions for a problem whose solutions
admit a dynamic programming-like recurrence relation, we
can still find a solution under additional
packing constraints, which need to be satisfied approximately. The number
of additional constraints can be very large, for example,
polynomial in the problem size.

Technically, we reinterpret the dynamic programming subproblems
and their solutions as a network design problem. Inspired
by techniques from, for example, the Directed Steiner Tree problem,
we construct a strong LP relaxation, on which
we then apply randomized rounding.


Our approximation guarantees on the packing constraints have roughly the form of a
$(n^{\epsilon} \mathrm{polylog}\ n)$-approximation in time $n^{O(1/\epsilon)}$, for any $\epsilon > 0$. By setting $\epsilon=\log \log n/\log n$, we obtain a polylogarithmic approximation in quasi-polynomial time, or by setting $\epsilon$ as a constant, an $n^\epsilon$-approximation in polynomial time.

While there are necessary assumptions on the form of the DP,
it is general enough to capture many textbook dynamic programs from Shortest Path to Longest Common Subsequence. Our algorithm
then implies that we can impose additional constraints on the solutions
to these problems.
This allows us to model various problems from the literature in approximation algorithms, many of which were not thought to be connected to dynamic programming. In fact,
our result can even be applied indirectly to some problems that
involve covering instead of packing constraints, for example, the Directed Steiner Tree problem, or those that do not directly
follow a recurrence relation, for example, variants of the Matching problem.

Specifically, we recover state-of-the-art approximation algorithms for Directed Steiner Tree and Santa Claus, and generalizations of them. We obtain new results for a variety of challenging optimization problems, such as Robust Shortest Path, Robust Bipartite Matching, Colorful Orienteering, Integer Generalized Flows, and more.

\end{abstract}

\section{Introduction}

\emph{Dynamic Programming} (DP) is a fundamental algorithmic principle that applies to problems that can be broken into smaller subproblems recursively, such that only a bounded number of \emph{different} subproblems can occur throughout the entire recursion tree. Then one
can solve the subproblems iteratively, which is the core idea of dynamic programming. Many textbook problems are solved in this way. Another important approach in algorithm design is the concept of \emph{Randomized Rounding}, which is widely used to transform a fractional solution to a linear programming (LP) relaxation into an actual solution. In this paper, we combine these two concepts into one framework. 
With this framework at hand, a great number of notorious problems in
approximation algorithms
can be solved by writing recurrences, which are sometimes
as easy as textbook dynamic programs.

We call the type of DP problems we consider \textit{Additive-DPs}, which we define in the following.
Since we will later apply randomized rounding, we linearize
our problems, namely, 
we assume that all solutions are encoded as non-negative
integer vectors $x\in\ZZ^\dimx$ and we want to minimize (or maximize) some linear objective $c\T x$ for some given $c\in\R^\dimx$.
There is a bounded family of subproblems $\mathcal I$, including the root problem $I^\circ$ itself.
Between the subproblems, there is a partial order $\prec$.

Each $I \in \mathcal I$ is associated with some set of feasible solutions $S(I) \subseteq \ZZ^\dimx$.
There is a set of subproblems $\mathcal {I}^\base\subseteq \mathcal I$ that we
call base problems.
For every base problem $I\in \mathcal{I}^\base$, we have either $S(I) = \{x^{(I)}\}$ for some given vector $x^{(I)} \in \ZZ^\dimx$ or $S(I) = \emptyset$. The latter indicates infeasibility. For every non-base problem $I \in \mathcal I \setminus \mathcal I^\base$
the solution set $S(I)$ is not given explicitly, but instead defined recursively as follows.
We can branch on some decision $C\in \{1,2,\dotsc, k_I\}$ for some $k_I \in \Z_{>0}$. Based on $I$ and $C$ we can decompose $I$ into a fixed
part $x^{(I, C)} \in \ZZ^\dimx$ and a recursive part corresponding to
$\ell(I,C)\ge 1$ subproblems $\Lambda_1(I, C),\dotsc,\Lambda_{\ell(I, C)}(I,C) \in \mathcal I$,
where $\Lambda_i(I,C)\prec I$ for all $i\in [\ell(I,C)]$.
The solution set for $I$ is
\begin{equation*}
    S(I) := \bigcup_{C = 1}^{k_I} \{x^{(I, C)}\} \oplus S(\Lambda_1(I, 1))\oplus \cdots \oplus S(\Lambda_{\ell(I,C)}(I, C)) \ ,
\end{equation*}
where $\oplus$ is the sumset operator, that is, $A\oplus B = \{a + b : a\in A, b\in B\}$.

A standard dynamic program to find
$x\in S(I^\circ)$ minimizing $c\T x$ would work as follows:
Starting with the base cases and then
following the order of $\prec$ it would compute for each subproblem
$I \in \mathcal I$ a solution $z(I)\in S(I)$ minimizing $c\T z(I)$.
Here, solutions to non-base cases are derived using
the recurrence above and
the fact that we can assume that the decomposed solutions
are minimizers of their respective subproblems.

A solution may directly or indirectly contain several solutions
to the same subproblem.
In a dynamic programming algorithm as above, one would always
obtain consistent solutions to these same subproblems. We
emphasize that in our semantics, we do not require this to be the
case. If for example $S(I)$ contains $S(I') \oplus S(I')$ then clearly $2x\in S(I)$ for all $x\in S(I')$, but we also have $x+x'\in S(I)$ for $x,x'\in S(I')$ with
$x\neq x'$. This is motivated by the fact that we will later add more constraints on top of the Additive-DP problem, in which case it is not clear anymore that the optimum solution should always take the same solution for many occurrences of the same subproblem. 

Consider now that, in addition to the recursive solution structure,
we need to satisfy
highly non-decomposable constraints. Namely, assume that we have packing constraints $Ax\le (1,\dotsc,1)\T$
for a given matrix $A\in [0,1]^{\nconstraints\times \dimx}$. We assume there is an optimal solution that satisfies these constraints. We are looking for a feasible solution to the Additive-DP problem, which is an $\alpha$-approximation with respect to the extra packing constraints. More precisely, we want a solution $x$ that satisfies  $Ax \le (\alpha,\dotsc,\alpha)\T$.

So, the Additive-DP problem is formally defined as follows:\footnote{Note that this definition no longer contains the option of infeasible base cases. Although they can be useful for simple modeling of problems, they can easily be removed in preprocessing and therefore we omit them for easier notation.}
\begin{tcolorbox}[colback=gray!10!white,colframe=gray!50!black,title=The Additive-DP Problem]
\textbf{Input:} Problems $\cal I$, root problem $I^\circ \in \cal I$, base problems $\mathcal{I}^\base \subseteq \mathcal{I} \setminus \{I^\circ\}$, relation $\prec$, \\[3pt]
\hspace*{34pt} $x^{(I)} \in \ZZ^d$ for every $I \in \mathcal{I}^\base$, \\[3pt]
\hspace*{34pt} $k_I \in \Z_{>0}$ for every $I \in \mathcal{I} \setminus \mathcal{I}^\base$,\\[3pt]
\hspace*{34pt} $x^{(I, C)} \in \ZZ^d, \ell(I, C) \in \Z_{\ge 1}$ for every $I \in \mathcal{I} \setminus \mathcal{I}^\base, C \in [k_I]$,\\[3pt]
\hspace*{34pt} $\Lambda_{i}(I, C)$ for every $I \in \mathcal{I} \setminus \mathcal{I}^\base, C \in [k_I], i \in [\ell(I, C)]$,\\[3pt]
\hspace*{34pt} $A \in [0, 1]^{\nconstraints \times \dimx}, c \in \R^{\dimx}$.\\[6pt]
\textbf{Output:} the solution $x \in S(I^\circ)$ satisfying $Ax \leq {\bf1}$ with the minimum $c^\rmT x$
\end{tcolorbox}


The main restriction of our model is that solutions are computed
using one large branching and the combination of subproblems
is via vector addition. In general dynamic programs,
one could imagine also more complicated circuits that compute
new solutions from previous ones.
Still, many textbook examples of dynamic programming, such as
the Knapsack problem, Shortest Path, or Longest Common Subsequence, can be cast in the framework above.
For instance, Shortest Path can be phrased as Additive-DP as follows. 

Without loss of generality, we assume the input graph $G = (V, E)$ is a DAG. A general graph can be converted to a DAG by creating $n$ layers of vertices, without changing the problem. Our goal is to find the shortest path from $s$ to $t$ in $G$, w.r.t.~the edge costs $c \in \R_{\geq 0}^E$.  We have one subproblem $I(v)$ for each vertex $v\in V$, with $I^\circ := I(s)$ being the root problem. The solutions are vectors in $\{0, 1\}^E$, with $S(I(v)):=\{\mathbb{1}_p: p\text{ is a path from $v$ to $t$}\}$ for each $v \in V$, where $\mathbb{1}_p$ is the indicator vector for the edges in $p$.  The relation $\prec$ is defined by the topological order of the vertices in $G$. The solutions can be defined recursively as follows. For the subproblem $I(t)$, we have $S(I(t)) = \{\bf0\}$. For every vertex $v \neq t$, we have 
 \begin{align*}
     S(I(v)) = \bigcup_{(v, u) \in \delta^+(v)} S(I(u)) \oplus \{\mathbb{1}_{(v, u)}\}, \text{ where $\mathbb{1}_{(v, u)}$ is the indicator vector for the edge $(v, u)$.}
 \end{align*}
Finding the shortest $s$-$t$ path is equivalent to finding $x\in S(I(s))$ minimizing $c\T x$. On top of this, we can define packing constraints on the $s$-$t$ path we choose, leading to the Robust $s$-$t$ Path Problem, which we discuss soon.  \medskip

If the number $\nconstraints$ of constraints is a constant or at most logarithmic, then one can hope to encode
the additional constraints in the subproblems of the DP and still
solve it (approximately) using pure dynamic programming. This approach can be seen for example in the recent work by Armbruster, Grandoni, Tinguely, and Wiese~\cite{armbr2026dp}, which is also on a general model for dynamic programs to which constraints are added. Their focus is on a logarithmic number of added constraints.
If $m$ is polynomially large, however, such approaches would lead
to an exponential blow-up of the number of subproblems; hence
dynamic programming is no longer feasible.

Inspired by techniques from network design problems, we present
reductions that simplify the structure of the dynamic programs. We then construct a non-trivial LP relaxation for optimizing over
the simplified DP structure while maintaining additional packing constraints. Before describing our result formally, we need to introduce the notion of a solution size, which will affect the running time and which we denote by $\size_I(x)$ for every $I \in \mathcal{I}$ and $x \in S(I)$. 
Intuitively, we can define a recursion tree for $x \in S(I)$ showing how $x$ is obtained, and $\size_I(x)$ is the minimum possible number of vertices over any such tree.
For a base problem $I$, we have $\size_I(x^{(I)}) = 1$. For any non-base problem $I$ and $x \in S(I)$, $\size_I(x)$ is defined as the minimum of $1 + \size_{\Lambda_1(I, C)}(z^{(1)}) + \size_{\Lambda_2(I, C)}(z^{(2)}) + \cdots + \size_{\Lambda_{\ell(I, C)}(I, C)}(z^{(\ell(I, C))})$, over all $C \in [k_I], z^{(1)} \in S(\Lambda_1(I, C)), z^{(2)} \in S(\Lambda_2(I, C)), \cdots, z^{(\ell(I, C))} \in S(\Lambda_{\ell(I, C)}(I, C))$ 
satisfying $x^{(I, C)} + z^{(1)}+z^{(2)}+\cdots+z^{(\ell(I, C))} = x$. 

\begin{theorem}
\label{thm:main}
 Suppose we are given an Additive-DP instance. Let 
 $\listslength := |\mathcal{I}_\base| + \sum_{I \in \mathcal{I}\setminus \mathcal{I}^\base, C \in [k_I]} (\ell(I, C)+1)$. Suppose we are given a promised upper bound $\sizebound$  on $\size_{I^\circ}(x^*)$ for the optimum solution $x^* \in S(I^\circ)$. 
 
Let $\epsilon > 0$. For some $\alpha = O\left(\frac{\sizebound^{\epsilon}}{\epsilon^2}\cdot \log \nconstraints\right) $, 
we can in time $(\sizebound\listslength)^{O(1/\epsilon)} \poly(\listslength,\sizebound,\nconstraints,\dimx)$ find a solution $x\in S(I)$ with $Ax \le \alpha \cdot {\bf1}$ and $c\T x \le \opt$, where $\opt$ is the optimum cost of the Additive-DP instance. When $c = \bf0$ and thus our goal is to find any $x\in S(I^\circ)$ satisfying the packing constraints, we can make $\alpha = O\left(\frac{\sizebound^{\epsilon}}{\epsilon}\cdot \log \nconstraints\right)$.
\end{theorem}

In many applications, all the parameters are  polynomially bounded in the input size $n$. Then our result implies a
polylogarithmic approximation in quasi-polynomial time by setting $\epsilon = \mathrm{loglog}(n)/\log(n)$, or a $O(n^{\epsilon})$-approximation
in $n^{O(1/\epsilon)}$-time by setting $\epsilon$ as a constant.
While our main theorem is on packing constraints and DP structure,
it can be applied indirectly to many other structures as well,
including certain covering problems, matching structures and more.
We will give an overview over the applications next.

\subsection{Applications to packing problems}

As a first application of Theorem \ref{thm:main} we give an algorithm to the Robust $s$-$t$ Path Problem. Indeed, this follows immediately from the subproblems and recurrence relation for Shortest Path that we outlined earlier. The result generalizes the work of \cite{li2024polylogarithmic}, which only gives the quasi-polynomial time polylogarithmic approximation for this problem, but no non-trivial approximation in polynomial time.
\begin{tcolorbox}[colback=gray!10!white,colframe=gray!50!black,title= The Robust $s$-$t$ Path Problem]
\textbf{Input:} A directed graph $G = (V, E)$ with edge costs $c \in \R_{\geq 0}^E$, start and destination vertices $s,t\in V$. \\[3pt]
\hspace*{32pt} $k$ length vectors $a^1,\dotsc,a^k \in [0, 1]^E$ on edges $E$. \\[3pt]
\textbf{Output:} A path $p$ in $G$ from $s$ to $t$ with $\sum_{e \in p}a^j_e\le 1$ for all $j \in [k]$, minimizing $\sum_{e \in p} c_e$.
\end{tcolorbox}
\begin{theorem}
    Let $\epsilon >0$. For some $\alpha = O(\frac{n^\epsilon}{\epsilon^2} \cdot \log k)$, we can in time $n^{O(1/\epsilon)}\cdot \poly(k)$ compute a path $p$ with $\sum_{e \in p}c_e \leq \opt$ and $\sum_{e \in p}a^j_e \leq \alpha$ for every $j \in [k]$, where $\opt$ is the optimum cost of a path that satisfies all the length requirements. The algorithm succeeds with high probability. 
\end{theorem} \medskip

Using our main theorem, we can also obtain interesting approximation results in cases where no non-trivial approximation was known before. Two examples are the Integer Generalized Flow Problem and Longest Common Subsequence with a bounded number of occurrences of each letter of the alphabet.

\begin{tcolorbox}[colback=gray!10!white,colframe=gray!50!black,title=The Minimum-Cost Integer Generalized Flow Problem]
\textbf{Input:} A directed graph $G = (V, E)$, a source $s\in V$ with excess $F\in \N$, a gain $g_e \in \ZZ$, cost $c_e\in \mathbb R_{\geq 0}$, and capacity $\capa_e\in \ZZ$ for each edge $e\in E$.

\textbf{Output:} An integer ``generalized flow'' $f: E \mapsto \ZZ$ satisfying
    \begin{equation*}
        \sum_{e\in\delta^+(v)} f(e) - \sum_{e\in\delta^-(v)} g_e \cdot f(e) = \begin{cases}
            F &\text{ if } v = s \\
            0 &\text{ if } v \in V\setminus \{s\} \ ,\\
        \end{cases}
    \end{equation*}
    \begin{equation*}
        f(e)\le \capa_e\ ,
    \end{equation*}
and minimizing 
\begin{equation*}
    \sum_{e\in E} c_e\cdot  f(e) \ .
\end{equation*}
\end{tcolorbox}
Generalized flow has been extensively studied in the literature in the fractional variants, see e.g.~\cite{tardos1998simple,olver2020simpler,jiang2025gen}, where one can obtain the optimal solution in polynomial time. Integral versions have also been studied in the past \cite{langley1973continuous} but no non-trivial approximation algorithm is known to the best of our knowledge. In our case, an $\alpha$-approximate Integer Generalized Flow will refer to a flow as above, but which violates the capacities by an $\alpha$ factor (i.e. $f(e)\le \alpha \cdot \capa_e$ for all $e\in E$).
\begin{theorem}\label{thm:gen-flow}
    Assume there exists a feasible integer generalized flow $f$ of cost $\opt$, then for any $\epsilon>0$, we can in time $(\|f\|_1\cdot n)^{O(1/\epsilon)}$ compute an $O\Big(\frac{(\|f\|_1\cdot n)^\epsilon}{\epsilon^2}\cdot \log n\Big)$-approximate integer generalized flow of cost at most $\opt$.
\end{theorem}
To prove this, we introduce a subproblem for each $s'\in V$, each excess $F'$, and each bound on the total flow,
where the solutions are the generalized flows with excess flow $F'$ in $s'$ that satisfy the bound on the total flow. The recurrence is straight-forward. We defer the details to Section~\ref{sec:app-packing}. \medskip

Next, we apply our theorem to the Bounded-Repetition Longest Common Subsequence Problem. This problem has been studied recently in~\cite{asahiro2020exact}, where it is shown to be APX-hard with respect to the approximation on the length of the common subsequence. We study it with respect to approximating the number of occurrences of each letter, which is the stronger setting, since it can easily be translated into an approximation algorithm for the length. To the best of our knowledge, no non-trivial approximation algorithm is known for either problem. There exists non-trivial approximation algorithms for related problems~\cite{asahiro2024polynomial}, but in most cases the approximation ratio is polynomial in the input size.

\begin{tcolorbox}[colback=gray!10!white,colframe=gray!50!black,title=The Bounded-Repetition Longest Common Subsequence Problem]
    \textbf{Input:} Two strings $a[1]a[2]\dotsc a[m]$ and
$b[1]b[2]\dotsc b[n]$ over alphabet $\Sigma$ and a bound $C\in\ZZ$.

\textbf{Output:} 
The longest common subsequence of both strings that contains
each character at most $C$ times.
\end{tcolorbox}
\begin{theorem}
    Assume that $n\ge m$
    and that there exists a common subsequence of $a$ and $b$ which does not contain any character more than $C$ times and is of length $\opt$. For any $\epsilon>0$, one can in time $n^{O(1/\epsilon)}$ compute a common subsequence of length at least $\opt$ which does not contain any character more than $C \cdot \Big(\frac{n^\epsilon}{\epsilon} \log |\Sigma|\Big)$ times.
\end{theorem}
The theorem follows immediately from reformulating the textbook DP
to our model. We defer the details to Section~\ref{sec:app-packing}.

\subsection{Applications to covering problems}
Our theorem is not limited to the use of additional packing constraints. 
Often, other constraints such as covering constraints can
be emulated using the DP structure and packing constraints.
As an example of this, we apply our result to a variant of Directed Steiner Tree,
which contains a covering constraint for the terminals.

\begin{tcolorbox}[colback=gray!10!white,colframe=gray!50!black,title=The Robust Directed Steiner Tree Cover Problem]
\textbf{Input:} A $n$-vertex directed graph $G=(V,E)$, a root $r \in V$, edge costs $c \in \R_{\geq 0}^E$, a bound $B \in \R_{\geq 0}$ on the cost, a set of terminals $K\subseteq V$, vectors $a^1, a^2, \cdots, a^k \in [0, 1]^E$.\\
\textbf{Output:} An out-arborescence $T$ with root $r$, $\sum_{e \in T}c_e \leq B$, $\sum_{e \in T}a^j_e \leq 1$ for every $j \in [k]$ that contains the maximum number of terminals in $K$. 
\end{tcolorbox}

\begin{theorem}
	\label{thm:DST}
	Let $\epsilon > 0$. For some $\alpha = O(\frac{n^\epsilon}{\epsilon^2} \log (k + |K|))$, we can in time $n^{O(1/\epsilon)}\poly(k)$ compute an out-arborescence $T$ satisfying $\sum_{e \in T}c_e \leq B$, $\sum_{e \in T}a^j_e \leq \alpha$ for every $j \in [k]$ that covers at least $\opt/ \alpha$ terminals, where $\opt$ is the number of terminals covered by the optimum solution that does not violate any constraint. The algorithm succeeds with high probability.
\end{theorem}

Ghuge and Nagarajan \cite{ghuge2022quasi} studied the Submodular Tree Orienteering problem. Compared to our problem, they consider the more general objective of maximizing $f(\{v \in T\})$ for a given submodular function $f:2^V \to \R_{\geq 0}$, but there are no packing constraints of the form $\sum_{e \in T}a^j_e \leq 1$. Moreover, they only prove a poly-logarithmic approximation in quasi-polynomial time.



We only sketch our approach here and defer details to Section~\ref{sec:app-covering}. The basic idea is that there is a subproblem $I(v, o, w)$ for each $v\in V$, $o\in \{0,1,\dotsc,\opt\}$ and $w\in \{0,1,\dotsc,n\}$, which intuitively represents trees rooted at $v$ with $o$ terminals and at most $w$ edges. However, as we cannot restrict the trees given by sub-problems to be edge-disjoint or have in-degree at most~$1$, we allow the solution to be any multigraph where all edges are reachable from $v$. This is not an issue as an inclusion-wise minimal solution is always an out-arborescence in $G$. It is easy to write a recurrence relation for solutions of these problems. 
 
 To prevent that we use a terminal many times, we add packing constraints that require each terminal to have at most one incoming edge. The packing constraints $\sum_{e \in T}a^j_e \leq 1, \forall j \in [k]$ carry over naturally to the Additive-DP problem.  Using Theorem~\ref{thm:main}, we can find an $\alpha = O\left(\frac{n^{\epsilon}}{\epsilon^2}\log(k + |K|)\right)$-approximate solution with cost at most $B$. This solution must cover $\opt / \alpha$ different terminals. \medskip

By repeating the process many times, we recover the state-of-the-art approximation for the Directed Steiner Tree (DST) problem given in \cite{charikar1999approximation,rothvoss2011directed,grandoni2019log2} up to a logarithmic factor: we obtain an $O\left(\frac{n^\epsilon}{\epsilon}\log^2 |K| \right)$-approximation in $n^{O(1/\epsilon)}$-time.\footnote{By losing a factor of $1+\epsilon$, we can assume costs are polynomially bounded integers. Then, we can encode the cost information into DP states, and apply Theorem~\ref{thm:main} without costs. This improves the factor of $\frac{1}{\epsilon^2}$ to $\frac{1}{\epsilon}$. However, our result is still worse than the state-of-the-art approximation of $O(\frac{n^\epsilon}{\epsilon}\log |K|)$ by a $\log |K|$ factor, as we need to take a union bound over $|K|$ packing constraints.}

Extensions of DST where soft bounds are imposed on the out-degree of vertices have been studied in~\cite{guo2022approximating}, where
a quasi-polynomial time polylogarithmic approximation algorithm is presented.
Using our framework it is straight-forward to capture degree bounds in the length functions, which matches the previous work up to logarithmic factors in quasi-polynomial time. We also give a polynomial time $O(n^\epsilon)$-approximation.
\medskip

In contrast to the Tree Orienteering problem, the traditional version of the Orienteering Problem considers $s$-$t$ walks instead of directed trees.
For a variant of this we can also give an approximation algorithm using similar ideas of emulating covering constraints. 
\begin{tcolorbox}[colback=gray!10!white,colframe=gray!50!black,title= The Colorful Orienteering Problem]
\textbf{Input:} A directed graph $G = (V, E)$ with edge costs $c \in \R_{\geq 0}^E$ and colors $\kappa \in [C]^E$ for some integer $C \geq 1$, a start $s\in V$, destination $t\in V$, and budget $B\in \R_{\geq 0}$. \\
\textbf{Output:} A walk $p$ (a path which is not necessarily simple) from $s$ to $t$ with $\sum_{e \in p}c_e\le B$ that
maximizes the number of visited colors, i.e.,	$|\{\kappa_e: e \in p\}|$, where the set is viewed as a normal set.
\end{tcolorbox}
\begin{theorem}
	\label{thm:colorful-orienteering}
    One can in time $n^{O(1/\epsilon)}$ compute an $O\left(\frac{n^\epsilon}{\epsilon^2}\log C\right)$-approximation for Colorful Orienteering for every $\epsilon > 0$. The algorithm succeeds with high probability.
\end{theorem}
The proof also follows from the fact that the Shortest Path Problem can be phrased as an Additive-DP problem, as explained earlier and by encoding covering constraints similar to the approach in Directed Steiner Tree. We defer the details to Section~\ref{sec:app-covering}. In~\cite{chekuri2005recursive} a logarithmic approximation in quasi-polynomial time is presented for the Colorful Orienteering Problem and a more general submodular variant. It was an open problem to get a non-trivial approximation in polynomial time~\cite{chekuri2024orienteering}.\notelr{I added this reference. Chandra mentioned this open problem to me explicitly during FOCS'24}

Similar to Directed Steiner Tree, our approach could also handle $k$ additional length functions on the edges of the path, and the violation of the packing constraints and the approximation ratio of the color coverage would become
$O\left(\frac{n^\epsilon}{\epsilon^2} \log (k + C)\right)$.

\subsection{Applications via augmentation}
For some problems, we cannot directly write
a recurrence that fits into the dynamic programming
framework, but we can still make indirect use of our theorem. The first example we give is the Robust Bipartite Perfect Matching Problem.

\begin{tcolorbox}[colback=gray!10!white,colframe=gray!50!black,title= The Robust Bipartite Perfect Matching Problem]
\textbf{Input:} A bipartite graph $G = (A\cup B, E)$ and $k$ length vectors $a^1,\dotsc,a^k \in [0, 1]^E$. \\[3pt]
\textbf{Output:} A perfect matching $M\subseteq E$ with $\sum_{e\in M} a^j_e\le 1$ for each $j\in [k]$.
\end{tcolorbox}

Intuitively, the reason why we can still apply our theorem
to this problem is that one can iteratively augment the matching via applications of the Robust $s$-$t$ Path algorithm.
Formally, one needs to be careful to have a small number
of iterations (by augmenting the matching by many edges at
every step), because we will introduce an error at each iteration.

For the Robust Bipartite Perfect Matching Problem, it was only known how to handle additional packing constraints in the general form if one would allow to output not a perfect matching but a version which loses a constant fraction of the edges. In that case, \cite{chekuri2011multi} can obtain an $O(\log n)$-approximation in polynomial time. Our result implies the following.
\begin{theorem}
\label{thm:matching}
Given an instance of the Robust Bipartite Perfect Matching Problem with $n$ vertices,
    for every $\epsilon > 0$, one can in time $n^{O(1/\epsilon)}\cdot\poly(k)$ compute a perfect matching that violates each length constraint by at most $O(\frac{n^\epsilon}{\epsilon^2} \cdot \log^3 (nk))$, assuming there exists a perfect matching satisfying all length constraints. The algorithm succeeds with high probability. 
\end{theorem}

Finally, we conclude this application section by one implication of our result for the Santa Claus problem.

\begin{tcolorbox}[colback=gray!10!white,colframe=gray!50!black,title=The Budgeted Santa Claus Problem]
\textbf{Input:} A set $P$ of $m$ players, a set $R$ of $n\ge m$ indivisible resources, resource $j$ gives value $v_{ij}$ to player $i$ at a cost of $c_{ij}$. There is a given budget $B$ for the total cost. \\
\textbf{Output:} An integral assignment $\sigma:R\mapsto P$ of resources to players such that $$ \min_{i\in P}\sum_{j:\sigma(j)=i}v_{ij}$$ is maximized and
the cost $\sum_{i,j:\sigma(j)=i}c_{ij}$ is at most $B$.
\end{tcolorbox} 
For the Budgeted Santa Claus 
problem, we can with a similar, yet more
sophisticated way reduce the problem to an augmentation problem,
which can be solved via the previous result on the Integer Generalized Flow Problem. The details are deferred to Section~\ref{sec:app-aug}.
Although this reduction is heavily based on previous works~\cite{chakrabarty2009allocating,bamas2025submodular}, we believe it is
technically less involved than those works.

\begin{theorem}\label{thm:santa}
    Suppose for the Budgeted Santa Claus Problem there exists a solution that gives each player a
    value of at least $V$ and the total assignment cost is at most $B$.
    One can in time $n^{O(1/\epsilon)}$ compute a bicriteria approximation, where
    the cost is at most $B\cdot n^{\epsilon}\mathrm{polylog}(n)$ and the value of each player is at least $V/ (n^{\epsilon}\mathrm{polylog}(n))$, for any $\epsilon > 0$.
\end{theorem}
This reproduces the state-of-the-art in~\cite{chakrabarty2009allocating} for the problem and
extends it to handle also costs.
Costs in the Santa Claus problem have also been considered in
the recent work~\cite{rohwedder2025cost}, but there in a much
more restricted variant where values cannot be arbitrary.

\subsection{Overview of our algorithm}
To devise an algorithm for Additive-DPs with additional packing constraints, we first reformulate the problem as an equivalent network design problem that we call the Flexible Tree Labeling Problem. 
In this problem we are given a set of labels and
the goal is to output a tree along with a labeling of its vertices.
The structure of the tree can be chosen arbitrarily, but
there are conditions on which labels of a parent can be combined with which labels of children. Furthermore, on the root and the leafs of the tree only specific labels can be used.
Over all used labels of the leafs packing constraints are imposed. The problem formulation contains some details that we leave out for simplicity and that are deferred to Section~\ref{sec:reduction_DPTree}.
In principle, the tree structure represents the recursion tree as in
a solution to the DP and the labels roughly correspond to branching decisions. With a careful definition the problem is exactly equivalent to Additive-DPs.

To prove Theorem~\ref{thm:main} we rely on two main techniques. The first one is to simplify the structure of the Flexible Tree Labeling Problem. Namely, we reduce it to the Perfect Binary Tree Labeling Problem, which is very similar except that the tree can no longer be chosen and must always be the perfect binary tree of a given height.
To achieve this, we observe that a tree can be recursively decomposed via vertex separators, such that after logarithmical depth the trees become singletons. Our Perfect Binary Tree Labeling instance mimics this recursive decomposition.
This uses similar ideas to~\cite{guo2022approximating}, where such an approach was used to solve degree-bounded Directed Steiner Tree variants.

We remark how the parameters transfer from the original Additive-DP instance to the Perfect Binary Tree labeling instance. The number $n$ of leaves in the perfect binary tree is $\poly(\sizebound)$, and size of the label set $L$ is $\poly(\sizebound, \listslength)$.\notelr{Etienne: placed this text here so that the explanation that follows has the correct parameters}

The second important step is to design a linear programming relaxation of Perfect Binary Tree Labeling. We now assume that the optimal tree has logarithmic height and every vertex except those in the last layer has two children. Assuming this, we can write a linear program in a recursive way which captures the problem. The recursion depth of this linear program is the same as the height of the tree, resulting in a size of $|L|^{O(\log n)}\poly(n, \nconstraints,\dimx)$, where $\nconstraints$ is the number of packing constraints, and $\dimx$ the dimension of the $x$ vectors in the original Additive-DP instance. Randomized rounding on this linear program yields the polylogarithmic approximation using standard Chernoff bounds. 
However, it remains to explain how to get the trade-off of an $O(n^\epsilon \cdot \poly\log \nconstraints)$-approximation in time $|L|^{O(1/\epsilon)} \poly(n,\nconstraints,\dimx)$, for example, in the regime where $\epsilon$ is constant. Since $n = \poly(\Delta)$, this gives roughly the claimed guarantee. Imagine the following process on the tree. We start at the root $r$ and collapse all the paths of length $\epsilon \log n$ from $r$ to some vertex $v$ into single edges $(r,v)$.
Then we repeat this process starting at each vertex $v$ connected directly to $r$ via such an edge, until reaching the leafs of the tree. Clearly this process transforms the binary tree of depth $\log n$ into a tree of depth $1/\epsilon$ and degree $n^\epsilon$. On this transformed tree, we can proceed the same way as before, writing a linear program of size $|L|^{O(1/\epsilon)}\poly(n,\nconstraints,\dimx)$ and use randomized rounding. However, because we have nodes of degree $n^\epsilon$, this creates additional correlations during the rounding. More precisely, we argue that the random variables of interest can be written as a sum of independent variables of value lying in the interval $[0,n^\epsilon]$. To apply the Chernoff bound we need to scale down variables by $n^\epsilon$; this is why $n^\epsilon$ appears in the approximation ratio.
\notelr{Etienne: slightly changed the parameters above to fit more closely our Theorem.}

Without costs, the recursive rounding algorithm gives a good concentration; this follows from the work of \cite{guo2022approximating}. With costs, the algorithm would maintain costs only in expectation.
However, to guarantee with probability~1 that our cost is at most the
cost of the LP solution, we combine the randomized rounding with ideas from~\cite{rohwedder2025cost}, namely we perform the randomized
rounding in several stages, where we can in each stage negate all choices. Then in each stage if the cost increases, we can negate the change leading to decreasing cost instead. If applied carefully,
this only increases the approximation ratio by a factor of $O(\frac1\epsilon)$, compared to the independent rounding procedure.

We remark that for most of the applications, if the costs $c$ are polynomially bounded integers, or we are allowed to ensure this by using a standard discretization procedure with a $(1+\epsilon)$-loss in approximation, then we can encode the cost information into the states of the DP. This will result in an Additive-DP instance without costs, and thus a simpler rounding algorithm and slightly better approximation guarantee.


\subsection{Related works}
Several notorious problems in approximation algorithm have a
state-of-the-art algorithm with a similar guarantee to our main theorem:
For both the Santa Claus problem and the Directed Steiner Tree problem, the best-known guarantees are an $O(n^\epsilon \text{polylog}(n))$-approximation in time $n^{O(1/\epsilon)}$ for any $\epsilon >0$ (see e.g. \cite{charikar1999approximation,grandoni2019log2,chakrabarty2009allocating}).
In some sense, our work is a very broad generalization of these two results.

For several other related problems, only the quasi-polynomial time
algorithm is known, but no non-trivial result in polynomial time. 
This includes the Submodular Orienteering problem studied in \cite{chekuri2005recursive}, where the elegant Recursive Greedy algorithm guarantees an $O(\log n)$-approximation in quasi-polynomial time.
Similar to our approach, the algorithm exploits the recursive structure of the problem.
This algorithm works roughly 
as follows: we branch on a decision $C$, then we recurse into the first subproblem and compute a solution that maximizes a given submodular objective. We then recurse in the next subproblem, but replace the original submodular function by the marginal gains based on the already selected solution from the first instance. We continue in this manner. Because of the added dependencies, this recursive algorithm
cannot be combined naively with dynamic programming. The depth of the recursion is in this case inherently logarithmic and the number of branches is polynomial, leading to only a quasi-polynomial time result.
In the special case of Colorful Orienteering, our algorithm can be applied and generalizes the $O(\log n)$-approximation to an $O(n^\epsilon \text{polylog}(n))$-approximation for any $\epsilon>0$, in time $n^{O(1/\epsilon)}$. A similar
status is in the Tree Orienteering~\cite{ghuge2022quasi} variant of Directed
Steiner Tree mentioned before
and in~\cite{li2024polylogarithmic}, which gives a quasi-polynomial time polylogarithmic approximation for the robust $s$-$t$ path problem, which we generalize to an $O(n^\epsilon \text{polylog}(n))$-approximation for any $\epsilon>0$, in time $n^{O(1/\epsilon)}$.

Another notable line of work related to our result revolves around the concept of dependent randomized rounding \cite{gandhi2006dependent,ageev2004pipage,chekuri2010dependent,chekuri2011multi}. Here, we are given a solution to an
LP relaxation, typically one that precisely captures the convex
hull of objects of interests, for example, matroid bases or b-matchings.
The goal is to randomly round it to one of these objects while
maintaining concentration bounds comparable to what one
would achieve in independent randomized rounding where the structure of the objects is disregarded.
In a typical applications one might obtain a fractional solution to the relaxation that satisfies additional packing constraints and then via the dependent rounding scheme, one can find a solution that still
satisfies these packing constraints approximately. Such an
application of dependent rounding is very similar to our use-cases.
Dependent rounding has strong limitations. It has mainly
been applied on matroids or with b-matchings (or assignment problems). In the case of
b-matchings, however, one has to make further restrictions. Namely,
the concentration holds only for linear functions over edges
that share an endpoint, see~\cite{chekuri2010dependent, gandhi2006dependent}.
It is difficult 
to apply the dependent rounding techniques to obtain a result
similar to Theorem~\ref{thm:matching}, since dependent rounding
is based on the naive LP relaxation and this relaxation has a large integrality
gap. For completeness we give an example of such an integrality gap in Appendix~\ref{sec:integrality-gap}. Thus, our techniques overcome a significant limitation of dependent rounding.


Another notable line of work related to our result is the field of robust optimization \cite{beyer2007robust}, where one often encodes uncertainty as different ``scenarios'', where each scenario corresponds to a linear objective. A particularly relevant paper is \cite{grandoni2014new}, which shows how to handle a constant number of additional budget constraints on top of well-known structures such as spanning tree or bipartite matching.

\subsection{Overview of the paper}
The rest of the paper is organized as follows. In Section~\ref{sec:reduction_DPTree} we give the reduction from Additive-DP to Perfect Binary Tree Labeling. In Section~\ref{sec:perfect-TL} we give our approximation algorithm for the Perfect Binary Tree Labeling problem, without costs. In Section~\ref{sec:costs}, we extend the previous algorithm to handle costs. Finally Section~\ref{sec:app-packing}, Section~\ref{sec:app-covering}, and Section~\ref{sec:app-aug}  are devoted to the applications of our main theorem.

\section{Reduction from Additive-DP to Perfect Binary Tree Labeling}
\label{sec:reduction_DPTree}

In this section, we rephrase the Additive-DP Problem as the Flexible Tree Labeling Problem, formally introduced Section~\ref{sec:flex} and then establish via reductions a simplified structure that resembles a perfect binary tree. This ultimately leads to the Perfect Binary Tree Labeling Problem, which is formally defined in Section~\ref{sec:perfect-TL}.  We shall use the following notations in this section. Given a rooted tree $T$ and a vertex $v$ in $T$, we use $\Lambda_T(v)$ to denote the set of children of $v$ in $T$, $\Lambda^*_T(v)$ to denote the set of descendants of $v$, including $v$ itself, and $T[v]$ to denote sub-tree of $T$ rooted at $v$. 

\subsection{Rephrasing Additive-DP as the Flexible Tree Labeling Problem}
\label{sec:flex}
Most of the time in the reduction, we work with the Flexible Tree Labeling (FTL) problem defined as follows. We are given a set $L$ of labels with a partial order $\prec$. There is a root label $\ell^\circ$, a set $L^\base \subseteq L \setminus \{\ell^\circ\}$ of base labels. Label $\ell^\circ$ is the unique maximal label, and all labels in $L^\base$ are minimal labels with respect to $\prec$.  For each base label $\ell \in L^\base$, we are given a non-zero vector $x^{(\ell)} \in \Z_{\geq 0}^\dimx$.   We are also given many allowed pairs of the form $(\ell \in L \setminus L^\base, L')$, where $L'$ is a non-empty multi-set of labels satisfying $\ell' \prec \ell$ for every  $\ell' \in L'$.

\begin{definition}
    A valid label tree to the instance, denoted as $T = (V_T, E_T, r, \ellv := (\ell_v)_{v \in V_T})$ is a rooted tree with vertices $V_T$, edges $E_T$, root $r \in V_T$, where every vertex $v \in V_T$ has a label $\ell_v \in L$. The following conditions must be satisfied:
    \begin{itemize}
    	\item $\ell_r = \ell^\circ$. 
        \item For every leaf $v \in V_T$, we have $\ell_v \in L^\base$.
    	\item For every non-leaf $v \in V_T$, $(\ell_v, \{\ell_u: u \in \Lambda_T(v)\})$ is an allowed pair. 
    \end{itemize}
    
    The solution vector of a valid label tree $T$ is defined as 
    \begin{align*}
        x^{(T)}:=\sum_{\text{leaf }v \in V_T}x^{(\ell_v)}.
    \end{align*}
\end{definition}

We are further given a vector $c \in \R^\dimx$ and a matrix $A \in [0, 1]^{\nconstraints \times \dimx}$.  Our goal is to find a valid label tree $T$ with the minimum $c\T x^{(T)}$ subject to the packing constraint $A x^{(T)} \leq {\bf1}$.

\begin{tcolorbox}[colback=gray!10!white,colframe=gray!50!black,title=The {\sc Flexible Tree Labeling} (FTL) Problem]
\textbf{Input:} Set $L$ of labels with partial order $\prec$, $\ell^\circ \in L, L^\base \subseteq L\setminus \{\ell^\circ\}$, $\left(x^{(\ell)} \in \Z_{\geq 0}^\dimx\right)_{\ell \in L^\base}$, allowed pairs of form $(\ell, L')$, $c \in \R^\dimx$, and $A \in [0, 1]^{\nconstraints \times \dimx}$. \\[3pt]
\textbf{Output:} A valid label tree $T = (V_T, E_T, r, \ellv:=(\ell_v)_{v \in V_T})$ satisfying $A x^{(T)} \leq {\bf1}$, so as to minimize $c^\mathrm{T} x^{(T)}$.
\end{tcolorbox}

Before converting the Additive-DP instance to an FTL instance, we need to address a small technical issue. In an Additive-DP instance, for each non-base problem $I$ and a choice $C \in [k_I]$, we may have a vector $x^{(I, C)}$. It is convenient for us to make $x^{(I, C)} = {\bf0}$. This can be done as follows: if $x^{(I, C)} \neq {\bf0}$, we introduce a new base problem $I'$ for the $(I, C)$ pair, define $x^{(I')} = x^{(I, C)}$, change $x^{(I, C)}$ to $0$, and add $I'$ to the children of $I$ for the choice $C$.

With the modification, we can then directly convert the Additive-DP instance to an equivalent FTP instance. Each sub-problem $I \in \mathcal{I}$ is a label in $L$. The base labels, $x$-vectors of base labels and the $\prec$ partial order are carried over directly from the Additive-DP instance. There is a unique label in $\ell^\circ$ correspondent to the root problem $I$. For a decomposition of a non-base instance $I$ into $\Lambda_1(I, C), \Lambda_2(I, C), \dotsc, \Lambda_{\ell(I, C)}(I, C)$ using a decision $C$, we let $(I, \{\Lambda_1(I, C), \Lambda_2(I, C), \dotsc, \Lambda_{\ell(I, C)}(I, C)\})$ be an allowed pair. \smallskip

Once we setup the initial FTL instance, we gradually refine it until we reach a Perfect Binary Tree Labeling instance.  In the initial instance, we have $|L^\base| + \sum_{\text{allowed pairs } (\ell, L')} |L'| \leq \listslength$, and the size of the optimum tree is at most $2\sizebound$, by our definition of $\sizebound$ and $\size_{I^\circ}(\cdot)$.

\paragraph{Making set size at most 2 in allowed pairs.} First, we can make the following guarantee: every allowed pair $(\ell, L')$ has $|L'| \in \{1, 2\}$. If some pair $(\ell, L')$ has $|L'| \geq 3$, we can apply the following operation.  We create a \emph{full binary tree} of labels with root label being $\ell$, the leaf labels being $L'$, and other labels being newly created; thus, we create $|L'| - 2$ new labels. We remove $(\ell, L')$ from the set of allowed pairs. For every non-leaf label $\ell'$, and its two child labels $\ell_1$ and $\ell_2$ in the tree, we let $(\ell', \{\ell_1, \ell_2\})$ be an allowed pair. The relation $\prec$ can be easily extended to include the new labels.  It is easy to see that this operation does not change the problem. 

After this operation, the number of labels now becomes at most $\listslength$. The size of the optimum tree becomes at most $4\sizebound$. We shall use $L$ to denote the label set of the new instance; so $|L| \leq \listslength$. abusing notations slightly, we still use $\sizebound$ to denote the upper bound on the size of the optimum tree.  As each allowed pair $(\ell, L')$ has $|L'| \leq 2$, there are at most $O(|L|^3)$ allowed pairs. We can disregard this bound from now on. 

\subsection{Making the height of a valid label tree small}
In this section, we create a new equivalent instance, where every valid label tree has height $O(\log \sizebound)$. 

\begin{definition}
Given a valid label tree $T = (V_T, E_T, r, (\ell_v)_{v \in V_T})$ for the original instance, a \emph{piece} of $T$ is a connected sub-graph $T'$ of $T$ (so $T'$ is also a rooted tree) satisfying the following condition:    
\begin{itemize}
    \item For every $v \in T'$, either $\Lambda_T(v) \subseteq T'$ or $\Lambda_T(v) \cap T' = \emptyset$.
\end{itemize}
We say some $v \in T'$ is a \emph{portal} in $T'$ if $v$ is a leaf in $T'$ but not $T$. 
\end{definition}
We shall consider a recursive procedure for decomposing a given valid label tree $T = (V_T, E_T, r, (\ell_v)_{v \in V_T})$ into a hierarchy of pieces. The procedure is defined in Algorithm~\ref{alg:decompose}; we call tree-decompose$(T)$ to decompose~$T$. 
\begin{algorithm}
    \caption{tree-decompose$(T')$\Comment{$T'$ is a piece of $T$}}
    \label{alg:decompose}
    \begin{algorithmic}[1]
        \If{$T'$ is a singleton} \Return \EndIf
        \If{$T'$ contains 1 level of edges} 
            \For{every leaf $v$ of $T'$} tree-decompose$(\{v\})$ \EndFor
            \State \Return
        \EndIf
        \State $D \gets $ set of portals in $T'$
        \If{$|D| \leq 2$} 
            \State $v \gets $ a non-root non-leaf vertex in $V_{T'}$ such that $T'[v]$ has size in $\Big[\floor{\frac{|V_{T'}|}3},\ceil{\frac{2|V_{T'}|}3}\Big]$
        \Else
            \State $v \gets $ a vertex in $V_{T'}$ such that $|T'[v] \cap D| = 2$ 
        \EndIf
        \State tree-decompose$\big(T' \setminus (\Lambda^*_{T'}(v) \setminus v)\big)$
        \State tree-decompose$(T'[v])$
    \end{algorithmic}
\end{algorithm}

The procedure naturally gives us a decomposition tree of pieces. Each node in the decomposition tree is a piece $T'$ in a recursion of tree-decompose. A piece $T'$ is a parent of another piece $T''$ if tree-decompose$(T')$ $T'$ calls tree-decompose$(T'')$ directly. 

\begin{lemma}\label{lem:dec-levels}
    The decomposition tree has $O(\log \sizebound)$ levels. Moreover, every piece $T'$ in the tree has at most $3$ portals. 
\end{lemma}
\begin{proof}    
    If $T'$ has at most $2$ portals, then each of its two child pieces has at most $3$ portals, as splitting $T'$ into two pieces increase the total number of portals by 1. If $T'$ contains $3$ portals, then by our choice of $v$ in Step~9, $T'[v]$ contains two portals of $T'$. Then each of its two child pieces has exactly $2$ portals.  This finishes the proof of the second part of the lemma.
    
    We say a recursion of tree-decompose is a size-reducing recursion if $v$ is defined in Step 7, and a portal-reducing recursion if $v$ is define in Step 9.  So, there cannot be two portal-reducing recursions, one being the parent of the other. A size-reducing recursion will reduce the size of the piece by a constant factor. Moreover, if $T'$ contains one level of edges, the recursion will terminate in one more level. Therefore, the recursion can run for at most $O(\log \sizebound)$ levels, finishing the proof of the lemma. 
\end{proof}
The decomposition tree will guide our construction of the tree labeling instance. For any piece $T'$ in the decomposition, we remember its root label, its size, and its portal labels in a meta-label. Replacing the pieces  with their meta-labels will give us a valid label tree in our new tree labeling instance, see Figure~\ref{fig:decompose} for an illustration. \medskip

\begin{figure}
    \centering
    \begin{tikzpicture}
        \node(u1)[draw] at (4, 4) {
    \begin{tikzpicture}[every node/.style={inner sep=2pt}]
         \node(v1)[circle, draw] at (5, 5) {1};
        \node(v2)[circle, draw] at (4, 4) {2};
        \node(v3)[circle, draw] at (6, 4) {3};   
        \node(v4)[circle, draw] at (5, 3) {4};   
        \node(v5)[circle, draw] at (4, 2) {5};   
        \node(v6)[circle, draw] at (6, 2) {6};   
        \draw[->] (v1) -- (v2);
        \draw[->] (v1) -- (v3);
        \draw[->] (v3) -- (v4);
        \draw[->] (v4) -- (v5);
        \draw[->] (v4) -- (v6);
        \node at (5, 1.5) {$(1, 6, \emptyset)$};
        \end{tikzpicture}
        };
        \node(u2)[draw] at (0, 3) {
    \begin{tikzpicture}[every node/.style={inner sep=2pt}]
         \node(v1)[circle, draw] at (5, 5) {1};
        \node(v2)[circle, draw] at (4, 4) {2};
        \node(v3)[circle, draw, dashed] at (6, 4) {3};   
        \draw[->] (v1) -- (v2);
        \draw[->] (v1) -- (v3);
        \node at (5, 3.5) {$(1, 3, \{3\})$};
        \end{tikzpicture}
        };
        \draw[->] (u1) -- (u2);
                \node(u3)[draw] at (8, 3) {
    \begin{tikzpicture}[every node/.style={inner sep=2pt}]
        \node(v3)[circle, draw] at (6, 4) {3};   
        \node(v4)[circle, draw] at (5, 3) {4};   
        \node(v5)[circle, draw] at (4, 2) {5};   
        \node(v6)[circle, draw] at (6, 2) {6};   
        \draw[->] (v3) -- (v4);
        \draw[->] (v4) -- (v5);
        \draw[->] (v4) -- (v6);
        \node at (5, 1.5) {$(3, 4, \emptyset)$};
        \end{tikzpicture}
        };
        \draw[->] (u1) -- (u3);
                \node(u4)[draw] at (-1.5, 0.5) {
    \begin{tikzpicture}[every node/.style={inner sep=2pt}]
        \node(v2)[circle, draw] at (4, 4) {2};
        \node at (4, 3.5) {$(1, 1, \emptyset)$};
        \end{tikzpicture}
        };
        \draw[->] (u2) -- (u4);
                        \node(u5)[draw] at (1.5, 0.5) {
    \begin{tikzpicture}[every node/.style={inner sep=2pt}]
        \node(v3)[circle, draw, dashed] at (6, 4) {3};   
        \node at (6, 3.5) {$(3, 1, \{3\})$};
        \end{tikzpicture}
        };
        \draw[->] (u2) -- (u5);
                        \node(u6)[draw] at (4, 0) {
    \begin{tikzpicture}[every node/.style={inner sep=2pt}]
        \node(v3)[circle, draw] at (6, 4) {3};   
        \node(v4)[circle, draw, dashed] at (5, 3) {4};   
        \draw[->] (v3) -- (v4);
        \node at (5.5, 2.5) {$(3, 2, \{4\})$};
        \end{tikzpicture}
        };
        \draw[->] (u3) -- (u6);
                \node(u7)[draw] at (9, -0.5) {
    \begin{tikzpicture}[every node/.style={inner sep=2pt}]
        \node(v4)[circle, draw] at (5, 3) {4};   
        \node(v5)[circle, draw] at (4, 2) {5};   
        \node(v6)[circle, draw] at (6, 2) {6};   
        \draw[->] (v4) -- (v5);
        \draw[->] (v4) -- (v6);
        \node at (5, 1.5) {$(4, 3, \emptyset)$};
        \end{tikzpicture}
        };
        \draw[->] (u3) -- (u7);
                                \node(u8)[draw] at (2, -2) {
    \begin{tikzpicture}[every node/.style={inner sep=2pt}]
        \node(v4)[circle, draw, dashed] at (5, 3) {4};   
        \node at (5, 2.5) {$(4, 1, \{4\})$};
        \end{tikzpicture}
        };
        \draw[->] (u6) -- (u8);
                \node(u9)[draw] at (6, -2) {
    \begin{tikzpicture}[every node/.style={inner sep=2pt}]
        \node(v5)[circle, draw] at (4, 2) {5};   
        \node at (4, 1.5) {$(5, 1, \emptyset)$};
        \end{tikzpicture}
        };
        \draw[->] (u7) -- (u9);
                \node(u10)[draw] at (12, -2) {
    \begin{tikzpicture}[every node/.style={inner sep=2pt}]
        \node(v6)[circle, draw] at (6, 2) {6};   
        \node at (6, 1.5) {$(6, 1, \emptyset)$};
        \end{tikzpicture}
        };
        \draw[->] (u7) -- (u10);
    \end{tikzpicture}
    \caption{Example of decomposition tree. Inside each rectangle is a piece on which tree-decompose is called. Portals are dashed. Labels are the numbers, with $2,5,6$ being base labels. The tripel below each piece is the corresponding label of the new FTL instance.}
    \label{fig:decompose}
\end{figure}
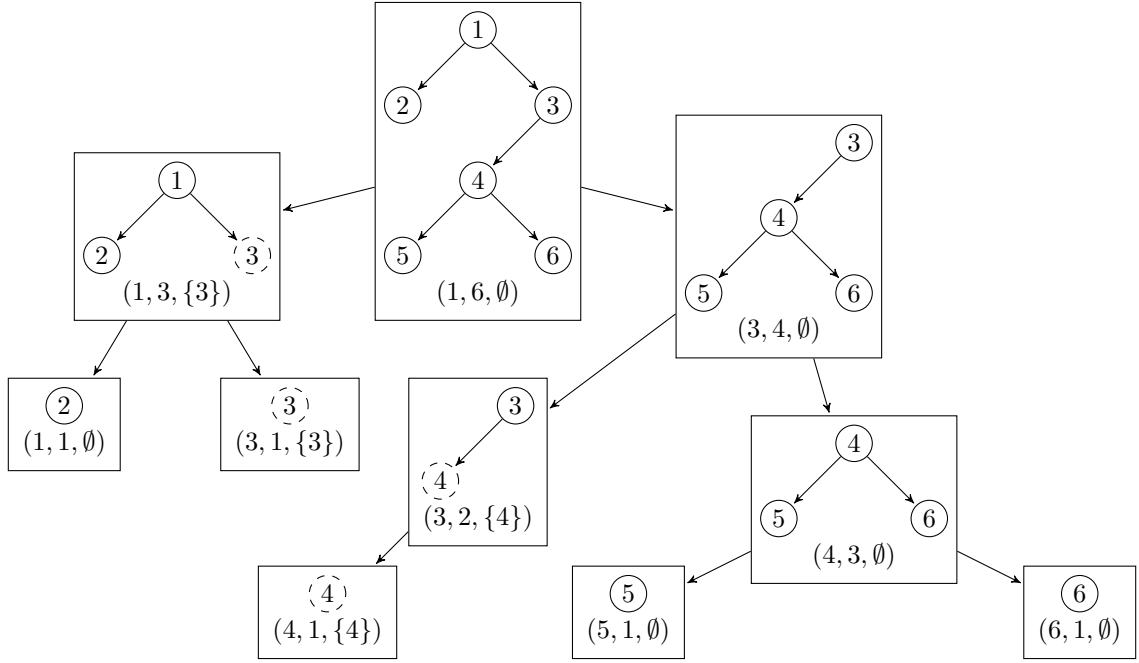

With the intuition, we now formally define the new FTL instance. The labels, root and base labels, the partial order $\prec$, and the $x$ vectors are defined as follows:
\begin{itemize}
    \item A label is of the form $(\ell, s, D)$, where $\ell \in L$, $s \in [\sizebound]$, and $D$ is a multi-set of labels in $L \setminus L^\base$ of size at most $3$, such that $\ell' \prec \ell$ for every $\ell' \in D$. There are two exceptions: $(\ell, 1, \{\ell\})$ for every $\ell \in L \setminus L^\base$ is also a label, and we have a root label $\ell^\bullet$.
    \item We have $(\ell', s', D') \prec (\ell, s, D)$ if $\ell' \prec \ell$, or $\ell' = \ell$ and $s' < s$.
    \item For every $\ell \in L^\base$, $(\ell, 1, \emptyset)$ is a base label with $x^{(\ell, 1, \emptyset)} = x^{(\ell)}$. For every $\ell \in L \setminus L^\base$, $(\ell, 1, \{\ell\})$ is a base label with $x^{(\ell, 1, \{\ell\})} = {\bf0}$. 
\end{itemize}
The intuitive meaning of a label $(\ell, s, D)$ is that the solution vectors
that can be obtained by trees rooted in it in the new instance correspond to the solution vectors that can be obtained from trees in the original instance rooted in label $\ell$, with $s$ nodes, and where for each $\ell'\in D$ we may 
have one leaf with label $\ell'$ (although normally this would not be allowed in a label tree, since it is not a leaf label).
We construct the set of allowed pairs as follows:
\begin{itemize}
    \item We allow $(\ell^\bullet, \{(\ell^\circ, s, \emptyset)\})$ for every $s \in [\sizebound]$.
    \item We allow $\big((\ell, s, D), \{(\ell, s', D'), (\ell'', s'', D'')\} \big)$ in the new instance if 
    \begin{itemize}
        \item $\ell, \ell' \in L$ with $\ell' \prec \ell$,
        \item $s\geq 3, s', s'' \in [2, s-1], s' + s'' = s + 1$, and
        \item $\ell'' \in D'$ and $D = (D' \setminus \{\ell''\}) \uplus D''$.
    \end{itemize}\notelr{rearranged since last two seem like corner cases}
    \item For some $\ell, \ell' \in L$ with $(\ell, \{\ell'\})$ allowed in the original instance, we allow $\big((\ell, 2, \{\ell'\} \setminus L^\base), \{(\ell', 1, \{\ell'\} \setminus L^\base)\}\big)$ in the new instance.
    \item For every $\ell, \ell', \ell'' \in L$ with $(\ell, \{\ell', \ell''\})$ allowed in the original instance, we allow $\big((\ell, 3, \{\ell', \ell''\} \setminus L^\base), \{(\ell', 1, \{\ell'\} \setminus L^\base), (\ell'', 1, \{\ell''\} \setminus L^\base)\} \big)$ in the new instance. 
\end{itemize}

\begin{lemma}
    For every valid label tree $T$ for the original instance of size at most $\sizebound$, there is a valid label tree for the new instance of height $O(\log \sizebound)$ with the same solution vector. 
\end{lemma}

\begin{proof}
    Consider the decomposition tree constructed by the procedure tree-decomposition$(T)$. Each node in the tree corresponds to a piece of $T$. Then, we replace each piece $T'$ with the label $(\ell, s, D)$, where
    \begin{itemize}
        \item $\ell$ is the label of the root of $T'$,
        \item $s$ is the size of $T'$, and
        \item $D$ is the set of labels of the portals of $T'$.
    \end{itemize}
    We take the resulting tree, and attach it to a root with label $\ell^\bullet$. We show this gives a valid label tree for the new instance, whose solution vector is equal to $x^{(T)}$. It has height $O(\log \Delta)$ by Lemma~\ref{lem:dec-levels}.
    
    First, the root piece has label $(\ell \in \ell^\circ, s = |V_T|, \emptyset)$, and $(\ell^\bullet,\{(\ell^\circ, s, \emptyset)\})$ is an allowed pair.  All leaf-pieces are of the form $(\ell, 1, \{\ell\} \setminus L^\base)$, which is a base label of the new instance.
    
    We then verify that the labels of an internal piece $T'$ and its child pieces form an allowed pair. Let $r'$ be the root of $T'$. If $T'$ contains 1 level of edges, then $T'$ is either an edge $(r', v')$, or two edges $(r', v'),(r', v'')$.  For the former case, the label for $T'$ is $(\ell_{r'}, 2, \{\ell_{v'}\} \setminus L^\base)$, and the label for the singleton $v'$ is $(\ell_{v'}, 1, \{\ell_{v'}\} \setminus L^\base)$. They form an allowed pair. For the latter case, the label for $T'$ is $(\ell_{r'}, 3, \{\ell_{v'}, \ell_{v''}\} \setminus L^\base)$, and the label for the singletons $v'$ and $v''$ are respectively $(\ell_{v'}, 1, \{\ell_{v'}\} \setminus L^\base)$ and $(\ell_{v''}, 1, \{\ell_{v''}\} \setminus L^\base)$. They also form an allowed pair. 

    Now assume $T'$ contains at least 2 levels of edges. Then tree-decompose$(T')$ chooses a non-leaf vertex $v \neq r'$ in $T'$. 
    $T' \setminus (\Lambda^*_v \setminus v)$ and $T'[v]$ are the two child-pieces of $T'$. 
    We have $\ell_v \prec \ell_{r'}$ and the total size of the two pieces is the size of $T'$ plus 1. $v$ is a portal in $T_1$.  The set of portals in $T$ is the set of portals in $T' \setminus (\Lambda^*_v \setminus v)$ excluding $v$, union the set of portals in $T'[v]$. Therefore, the labels for the three pieces form an allowed pair. 

    Finally we show that the resulting label tree has solution vector equaling to $x^{(T)}$. Each non-root $v$ in $T$ corresponds to a singleton in the decomposition tree. For a leaf $v$ with a label $\ell:= \ell_v \in L^\base$, the singleton in the decomposition tree has label $(\ell, 1, \emptyset)$ and we have $x^{(\ell, 1, \emptyset)} = x^{(\ell)}$. The $x$ vector for a singleton $v$ for a non-leaf $v$ is $0$. This finishes the proof. 
\end{proof}

\begin{lemma}
    Let $B = (V_B, {E_B}, r_B, ({\bfl_p})_{p \in V_B})$ be a valid label tree for the new instance. There is a valid label tree $T$ for the original instance with size at most $\sizebound$ and $x^{(T)} = x^{(B)}$. 
\end{lemma}
\begin{proof}
    We recover the label tree $T$ for the original instance from the bottom to top of $B$. First, we remove the root,  whose label is $\ell^\bullet$, from $B$. We use $r_B$ to denote the new root of $B$. For every $p \in V_B$, we create a rooted tree $T^p$ with labels on vertices; labels are in $L$. Some leaves in $T^p$ are \emph{marked}, and the others are \emph{unmarked}. Let $\bfl_p = (\ell_p, s_p, D_p)$. We guarantee the following conditions for $T^p$:
    \begin{itemize}
        \item The root of $T^p$ has label $\ell_p$.
        \item The size of $T^p$ is $s_p$.
        \item The set of labels of marked vertices in $T^p$ is $D_p$.
        \item If a leaf is not marked, then its label is in $L^\base$.
        \item For any non-leaf node $v$ in $T^p$, its label and the label set of its children form an allowed pair in the original instance.
    \end{itemize}

    Now we describe how to construct $T^p$'s for $p$ from bottom to top of $B$. 
    
    Assume $\bfl_p = (\ell_p = \ell, s_p = 1, D_p)$. Then $p$ must be a leaf of $B$, and $D_p = \{\ell\} \setminus L^\base$.  Then, $T^p$ contains a singleton with label $\ell$, and the vertex is marked if and only if $\ell \notin L^\base$. Clearly all the 5 conditions are satisfied.

    Then assume $p$ is the parent of 1 leaf $p'$ in $B$. Then, $T^{p'}$ is a singleton, $s_p = 2$ and $D_p = D_{p'}$. $T^{p}$ is constructed by creating a root vertex with child $T^{p'}$; the root has label $\ell_p$. All the 5 conditions are satisfied. In particular, the last condition is satisfied as $(\ell_p, \{\ell_{p'}\})$ is an allowed pair in the original instance. 

    Similarly, if $p$ is the parent of 2 leaves $p'$ and $p''$ in $B$, $T^p$ is constructed by creating a root vertex with label $\ell_p$, and two child singletons $T^{p'}$ and $T^{p''}$. 

    It remains to consider the case where $p$ has two children $p'$ and $p''$ in $B$, $\ell_{p''} \prec \ell_{p} = \ell_{p'}$, $s_p\geq 3, s_{p'}, s_{p''} \in [2, s_p-1]$, $s_{p'} + s_{p''} = s_p + 1$, $\ell_{p''} \in D'$ and $D_p = (D_{p'} \setminus \{\ell_{p''}\}) \uplus D_{p''}$. To construct $T^p$, we take $T^{p'}$, and any leaf $v$ in $T^{p'}$ with label $\ell_{p''}$, which exists by that $\ell_{p''} \in D_{p'}$, identify the root of $T^{p''}$ with $v$; as $v$ is not a leaf anymore, it is unmarked.  All the 5 properties are satisfied: the third property holds as $D_p = (D_{p'} \setminus \ell_{p''}) \cup D_{p''}$, the fourth and fifth properties follow from the respective properties for $T^{p'}$ and $T^{p''}$.     

    Notice that the final tree $T^{r_B}$ created has no marked leaves. Therefore it is a valid label tree of size $s_{r_B} \leq \sizebound$. There is a one-to-one correspondence between non-root vertices of $T^{r_B}$ and the leaf vertices of $B$. For a leaf vertex $v$ of $T^{r_B}$, we have $\ell_v \in L^\base$, and the correspondent vertex in $T^{r_B}$ has label $(\ell_v, 1, \emptyset)$ and $x^{(\ell_v, 1, \emptyset)} = x^{\ell_v}$. For non-leaf non-root vertex $v$ of $T^{r_B}$, we have $\ell_v \notin L^\base$, and the correspondent vertex in $T^{r_B}$ has label $(\ell_v, 1, \{\ell_v\})$ and $x^{(\ell_v, 1, \{\ell_v\})} = 0$. Therefore, we have $x^{(B)} = x^{(T)}$. 
\end{proof}

\subsection{Final cleaning}
Again we use $L, \ell^\circ, L^\base$ denote the labels, root and base labels of the current FTL instance. It almost satisfies the properties of a Perfect Binary Tree Labeling
instance that we shall define in the next section, except that we need to address the issue that a valid label tree might not be a perfect binary tree. 

To address the issue, we can introduce some dummy labels and add dummy vertices to a valid label tree.  Let $H = O(\log \sizebound)$ be an upper bound on the height of the tree, i.e, the maximum number of edges in a root-to-leaf path in the tree. The new labels will be of the form $(h, \ell)$ where $h$ is an integer in $[0, H]$ and $\ell \in L \cup \{\bot\}$. Intuitively, for a label $(h,\ell)$,
$h$ describes how many levels are left and it is used like a counter. If the tree would end before level $H$,
we instead add dummy labels. Similarly, if a node has only one child, we add
a tree of dummy labels as the other child. Formally:
\begin{itemize}
    \item For every $h \in [H]$ and an allowed pair $(\ell, \{\ell', \ell''\})$ in the original instance, we allow $\big((h, \ell), \{(h-1, \ell'), (h-1, \ell'')\}\big)$ in the new instance.
    \item For every $h \in [H]$ and an allowed pair $(\ell, \{\ell'\})$, we allow $((h, \ell), \{(h-1, \ell'), (h-1, \bot)\})$ in the new instance.
    \item For every $h \in [H]$ and a base label $\ell$, we allow $((h, \ell), \{(h-1, \ell), (h-1, \bot)\})$ in the new instance.
    \item For every $h \in [H]$, we allow $((h, \bot), \{(h-1, \bot), (h-1, \bot)\})$ in the new instance.
\end{itemize}

Finally, the base labels of the new instances are $(0, \ell), \ell \in L \cup \{\bot\}$.  Let $x^{(0, \ell)} = x^{(\ell)}$ if $\ell \in L^\base$, and $x^{(0, \ell)} = 0$ if $\ell = \bot$. The root label is $(H,\ell^\circ)$ \medskip

We summarize the properties of the FTL instance we obtained. Every valid label tree of the instance is a perfect binary tree of height $O(\log \sizebound)$.
Recall that in the given Additive-DP instance, we have $\sizebound$ is the promised size of the optimum solution, and $\listslength = |\mathcal{I}^\base| + \sum_{I \in \mathcal{I} \setminus \mathcal{I}^\base, C \in [k_I]}(\ell(I, C) + 1)$.
The total number of labels in the FTL instance is at most $O(\sizebound\cdot\log(\sizebound)\cdot\listslength^4)$.   As each allowed pair $(\ell, L')$ now has $|L'| = 2$, we only have at most $|L|^3$ different allowed pairs.

\section{Perfect Binary Tree Labeling without costs}
\label{sec:perfect-TL}
As we showed
before, there are approximation-preserving reductions from Additive-DP to Flexible Tree Labeling and from the latter to Perfect Binary Tree Labeling, the 
variant where the structure of
the output tree is fixed to the perfect binary tree of a given height. 
In this section we give an LP-based approximation algorithm for
the perfect binary tree case and therefore also for Additive-DP, but without costs.
The case without costs is simpler and introduces many of the central
techniques, such as the LP relaxation. It also has a slightly better approximation
guarantee for the packing constraints.
In Section~\ref{sec:costs}, we then present a variant that preserves costs
perfectly.

For convenience we recall the setting for the Perfect Binary Tree Labeling problem, as we shall use slightly different terminology compared to the Flexible Tree Labeling problem.   We are given a perfect binary tree $T = (V, E)$ with root $r\in V$ and height $H$, and $n:=2^H$ leaves. Furthermore, we are given a set $L$ of labels. For each label $\ell \in L$, we are given a vector $x^{(\ell)} \in \Z_{\geq 0}^\dimx$. We are also given a set $\Gamma \in L \times L \times L$ of triples of the form $(\ell_{\rm P},\ell_{\rm L},\ell_{\rm R})$ that indicate that $(\ell_{\rm P}, \{\ell_{\rm L}, \ell_{\rm R}\})$ is an allowed pair. We assume without loss of
generality that the left child always has the second label $\ell_{\rm L}$ and the right child the third label $\ell_{\rm R}$. Finally, we are given a label $\ell^\circ\in L$ to use with the root. 

Notice that in Perfect Binary Tree Labeling, we drop the distinction into root, leaf, and other labels
and the partial order on labels.
Instead, we only force the root to have the specific label $\ell^\circ$. 
With a fixed tree, one can easily use the allowed pairs to enforce
a partial order and that leafs can only take specific labels.

A valid labeling of the binary tree, denoted by $\ellv = (\ell_v)_{v \in V}$ is such that the following conditions are satisfied:
\begin{itemize}
	\item $\ell_r = \ell^\circ$. 
	\item For every non-leaf $v \in V$ with left child $u\in V$ and right child $w\in V$, we have $(\ell_v, \ell_u, \ell_w) \in \Gamma$. 
\end{itemize}
The solution vector of a valid labeling is defined as 
\begin{align*}
    x^{(\ellv)}:=\sum_{\text{leaf }v \in V}x^{(\ell_v)}.
\end{align*}

We are further given a vector $c \in \R_{\geq 0}^\dimx$ and a matrix $A \in [0, 1]^{\nconstraints \times \dimx}$.  Our goal is to find a valid labeling $\ell$ with the minimum $c\T x^{(\ellv)}$ subject to the packing constraint $A x^{(\ellv)} \leq {\bf1}$.

\begin{tcolorbox}[colback=gray!10!white,colframe=gray!50!black,title=The {\sc Perfect Binary Tree Labeling} Problem]
\textbf{Input:} Perfect binary tree $T = (V_T, E_T)$ with height $H$ and $n = 2^H$ leaves, root $r\in V_T$, set $L$ of labels, $\ell^\circ\in L$, $\left(x^{(\ell)} \in \Z_{\geq 0}^\dimx\right)_{\ell \in L}$, set $\Gamma \subseteq L \times L \times L$, and $A \in [0, 1]^{\nconstraints \times \dimx}$. \\[3pt]
\textbf{Output:} A valid labeling $\ellv := (\ell_v)_{v \in V}$ of the given binary tree satisfying $\ell_r = \ell^\circ$ and $A x^{(\ellv)} \leq {\bf1}$, so as to minimize $c\T x^{(\ellv)}$.
\end{tcolorbox}
In this section, we present the feasibility variant that ignores the costs: 
\begin{theorem}
\label{thm:tree_labeling_approx}
    For every $\epsilon > 0$, there is a randomized multi-criteria approximation algorithm for Perfect Binary Tree Labeling
    that in time $|L|^{O(1/\epsilon)}\cdot \poly(n, \nconstraints, \dimx)$ finds 
    a labeling with corresponding solution vector $x^{(\ellv)}$ such that:
    \begin{itemize}
        \item $Ax^{(\ellv)} \le O\left(\frac{n^\epsilon}{\epsilon} \log \nconstraints\right) \cdot {\bf1}$,
    \end{itemize}
    assuming there exists a labeling $\ellv^*$ with solution vector $x^{(\ellv^*)}$ such that 
     $Ax^{(\ellv^*)} \le {\bf 1}$.
\end{theorem}

\subsection{Collapsing the tree $T$}
    Assume without loss of generality that $1/\epsilon \in \N$
    by reducing $\epsilon$ slightly and assume that $\epsilon H\in \N$ by adding at most $1/\epsilon$ dummy layers to the tree. More precisely, replace every leaf vertex $v$ by a small perfect binary tree of the correct depth. Inside this dummy tree, we create dummy labels that enforce that exactly one leaf in the dummy tree inherits the vector that $v$ had, while other leaves get the vector $\textbf{0}$. This is similar to the postprocessing at the end of Section~\ref{sec:reduction_DPTree}.
    
    The latter increases the number of leafs only by a factor of $2^{1/\epsilon}$ and the number of labels by a constant,
    both of which is insignificant towards our goals.
    After this normalization, $H$ is a multiple of $\epsilon H$.

    We work with a small number of $1/\epsilon + 1$ ``super-layers''.
    Towards this, let us denote by $V^{(i)}$ the set of vertices at depth $\epsilon H \cdot i$ in the tree.
    Then $V^{(0)}$ consists only of the root and $V^{(1/\epsilon)}$ are all leafs.
    We define
    $V^{(\le i)} = V^{(0)}\cup\cdots \cup V^{(i)}$. Note that $V^{(\le 1/\epsilon)}$ does not contain all vertices $V$, but only those in super-layers.
    
    For $v\in V^{(i)}$, $i < 1/\epsilon$,
    let $\Lambda^{+}_v \subseteq V^{(i+1)}$ be the descendants
    of $v$ in $V^{(i+1)}$.
    Note that $|\Lambda^{+}_v| = n^\epsilon$.
    Let $I_v^{+}$ be all vertices in layers $\epsilon H \cdot i+1,\dotsc,\epsilon H \cdot (i+1) - 1$ that are descendants of $v$.
    In other words, $I_v^{+}$ are all inner vertices on paths between $v$ and a vertex of $\Lambda^{+}_v$.
    
    Given some label $\ell\in L$, let ${\cal L}^+_{v,\ell}$ be the set of all label tuples that the vertices $\Lambda^+_v$ could have, given $v$ is labeled with $\ell$. More formally,
    \begin{align*}
        {\cal L}^+_{v,\ell} = \{ (\ell_u)_{u\in \Lambda^{+}_v} :\  
         & \ell_w\in L \text{ for all } w\in \{v\} \cup I_v^{+} \cup \Lambda^+_v \\
        &\text{s.t. } (\ell_w, \ell_{w'}, \ell_{w''}) \in \Gamma \text{ for all } w\in\{v\}\cup I_v^{+} \text{ with children } w', w'' \\
        &\text{and } \ell = \ell_v \} 
    \end{align*}

    \subsection{Relaxation for Perfect Tree Labeling}
    We write a continuous relaxation in a recursive manner: $\LP^{(v, \ell)}$ denotes the relaxation for the problem rooted at a super-layer vertex $v\in V^{(\le 1/\epsilon)}$, where the label of $v$ is fixed to $\ell$.  The way we write the linear program may be unusual, but
    this form is cleaner and will be more convenient for the rounding algorithm.
    We will give an equivalent formulation that
    shows that it can be implemented as a linear program of size $|L|^{O(1/\epsilon)}\cdot \poly (n,m,d)$ later.

Given a valid labeling $\ellv$, we define 
\begin{align*}
	x^{(\ellv, v)} = \sum_{\text{leaf descendants $u$ of $v$}} x^{(\ell_u)}
\end{align*}
as the sum of vectors for the leaf descendants of $v$. 

$\LP^{(v, \ell)}$ tries to capture the convex hull of $x^{(\ellv, v)}$, given that $\ell_v = \ell$. For this purpose it contains variables $x\in \RR^d$. In the correspondent integer program, $x$ is the vector $x^{(\ellv, v)}$ for some labeling $\ellv$ with $\ell_v = \ell$. Furthermore, for every $u \in \Lambda^+_v$ and $\ell' \in L$, there is a variable $\chi_{u, \ell'}$, which we think of as $1$ if the label of $u$ is $\ell'$ and $0$ otherwise. 
Finally, we use variables $x^{(u, \ell')}\in \RR^d$, which we think of to be $x^{(\ellv, u)}$ in the case that $\ell_u = \ell'$ and arbitrary otherwise.
These variables allow us to recursively build up $x$ from the solution vectors for vertices in $\Lambda^+_v$.

We also use $\LP^{(v, \ell)}$ to denote the polyhedron of feasible solutions projected to the coordinates of the $x$-vector. 
With this at hand, we can define $\LP^{(v, \ell)}$ recursively. For a leaf $v$ and $\ell \in L$, $\text{LP}^{(v, \ell)}$ contains the single vector $x^{(\ell)}$. Now, focus on an internal vertex $v \in V$ and $\ell \in L$. $\text{LP}^{(v, \ell)}$ is defined as follows:
\begin{align}
    Ax &\leq {\bf 1} \label{LPC-recursive:packing}\\
    x &= \sum_{u \in \Lambda^+_v, \ell'} \chi_{u, \ell'} \cdot x^{(u, \ell')} \label{LPC-recursive:x}\\
    \left(\chi_{u, \ell'}\right)_{u \in \Lambda^+(v), \ell' \in L} &\in \conv({\cal L}^+_{v, \ell}) \label{LPC-recursive:convex}\\
    x^{(u, \ell')} &\in \text{LP}^{(u, \ell')} \text{ or } \chi_{u,\ell'} = 0 &\qquad &\forall u \in \Lambda^+_v, \ell' \in L \label{LPC-recursive:recurse}
\end{align}
    
In the LP,  \eqref{LPC-recursive:packing} are the packing constraints and \eqref{LPC-recursive:x} describes that we can decompose $x^{(\ellv,v)}$ into the sum of $x^{(\ellv,u)}$, $u\in \Lambda_v^+$, \eqref{LPC-recursive:convex} requires that the vector of $\chi$-variables is in $\conv({\cal L}^+_{v,\ell})$, which is defined as follows:
    \begin{equation*}
        \conv({\cal L}^+_{v,\ell}) := \conv\left(\left\{\chi(\ellv) : \ellv \in {\cal L}^+_{v,\ell} \right\}\right).
    \end{equation*}
     Here $\conv(\cdot)$ denotes the convex hull and $\chi(\ellv)$ is a $(|\Lambda^+_v| \cdot |L|)$-dimensional vector defined as follows: For every $u \in \Lambda^+_v$ and $\ell' \in L$, $\chi(\ellv)_{u, \ell'}$ indicates whether $\ell_u = \ell'$.
    Finally,  \eqref{LPC-recursive:recurse} requires each $x^{(u, \ell')}$ to be in the polyhedron $\LP^{(u, \ell')}$. It may be possible that $\LP^{(u, \ell')} = \emptyset$ for some $u \in \Lambda^+_v$ and $\ell' \in L$, in which case there is no valid $x^{(u, \ell')}$. However, we allow this to happen, as long as $\chi_{u, \ell'} = 0$.

It may not be obvious to see that the previous linear program is solvable in time $|L|^{O(1/\epsilon)}\cdot \poly (n,m,d)$, or that it even is a linear program. For this, we translate it to an equivalent linear program that has size $|L|^{O(1/\epsilon)}\cdot \poly (n,m,d)$ in Section~\ref{sec:explicit_LP}. Before that, we need to give a polynomial size extended formulation of $\conv({\cal L}^+_{v,\ell})$. This is done in Section~\ref{sec:extended_conv}. 

    \subsection{Polynomial-size description of $\conv({\cal L}^+_{v, \ell})$}
    \label{sec:extended_conv}
       In this section, we give a polynomial size extended formulation of $\conv({\cal L}^+_{v,\ell})$. Let $\Gamma^+(\ell') \subseteq \Gamma$ be the set of all triples $(\ell_{\rm P}, \ell_{\rm L}, \ell_{\rm R})\in\Gamma$ with $\ell_{\rm P} = \ell'$.
    For every $u\in I^+_v \cup \Lambda^+_v$ let $\Gamma^-(u, \ell') \subseteq \Gamma$
    denote the set of triples $(\ell_{\rm P}, \ell_{\rm L}, \ell_{\rm R}) \in \Gamma$
    where $\ell_{\rm L} = \ell'$ and $u$ is a left child or $\ell_{\rm R} = \ell'$
    and $u$ is a right child.
    Consider the following linear program, which assigns each vertex in $\{v\} \cup I^+_v$ a triple, enforces consistency between
    these assignments and derives the labels of $\Lambda^+_v$ from them.
    
    For every $u \in \{v\} \cup I^+_v$ and triple $t \in \Gamma$, we have a variable $\phi_{v, t}$ indicating if the labels of $v$ and its two children form the triple $t$. For every $u \in \Lambda^+_v$ and label $\ell' \in L$, we have a variable $\chi_{u,\ell'}$ indicating if the label of $u$ is $\ell'$. Consider the following LP:
     \begin{align}
      \sum_{t \in \Gamma^+(\ell)} \phi_{v, t} &= 1 \quad & \label{LPC-convex:start} \\
        \sum_{t\in \Gamma^+(\ell')} \phi_{u,t} &= \sum_{t\in \Gamma^-(u,\ell')} \phi_{w, t} \quad &\forall u\in I^+_v, w\in\{v\} \cup I^+_v, u\text{ child of } w, \ell' \in L \\
         \chi_{u,\ell'} &= \sum_{t\in\Gamma^-(u,\ell')} \phi_{w,t}  &\forall u\in \Lambda^+_v, w\in \{v\}\cup I^+_v, u\text{ child of } w \\
        \phi_{v, t} &= 0\quad &\forall t \in \Gamma\setminus \Gamma^+(\ell) \\
        \phi_{u, t} &\ge 0\quad &\forall u\in \{v\}\cup I^+_v, t\in \Gamma  \label{LPC-convex:end}
	\end{align}

    Clearly the LP above is a relaxation of $\conv({\cal L}^+_{v,\ell})$. We will prove that the projection of the above LP onto $\chi:=(\chi_{u,\ell'})_{u\in \Lambda^+_v,\ell'\in L}$ is in fact equal to $\conv({\cal L}^+_{v,\ell})$. To do this, we design a simple randomized rounding algorithm which rounds a solution $(\chi, \phi)$ to an integral solution $\ellv$ such that $\mathbb P[\ell_u = \ell'] = \phi_{u,\ell'}$ for all $u\in \Lambda^+_v,\ell'\in L$.

    The rounding proceeds recursively from $v$ to $\Lambda^+_v$. The recursive procedure is parameterized by a vertex $u \in \{v\} \cup I^+_v \cup \Lambda^+_v$; at the beginning, we let $\ell_v = \ell$ and we call the procedure for $v$.   In the procedure for $u$, we do nothing if $u \in \Lambda^+_v$. So we assume $u \in \{v\} \cup I^+_v$. 
    %
    We select a triple  $t\in \Gamma^+(\ell_u)$. The triple $t$ is selected with probability equal to 
    \begin{equation*}
        \frac{\phi_{u, t}}{\sum_{t'\in \Gamma^+(\ell_u)} \phi_{u, t'}}\ .
    \end{equation*}
    Assume $u'$ and $u''$ are the left and right children of $u$. Then we define $\ell_{u'}$ and $\ell_{u''}$ so that $(\ell_u, \ell_{u'}, \ell_{u''}) = t$. We recursively call the procedure for $u'$ and $u''$. 
    \begin{lemma}
        \label{lem:rounding_pi}
        The naive randomized rounding of a solution $(\chi, \phi)$ ensures that $\mathbb P[\ell_u = \ell'] = \sum_{t\in \Gamma^-(u, \ell')}\phi_{w, t}$ for all $w \in \{v\} \cup I^+_v,  u\in I^+_v \cup \Gamma^+_v$ being a child of $w$, and $\ell' \in L$.
    \end{lemma}
    \begin{proof}
        We prove this statement by induction from the top to bottom. Consider the base case where $w = v$ and $u$ is a child of $v$.
        \begin{equation*}
            \mathbb P[\ell_u = \ell'] = \sum_{t\in \Gamma^{-}(u,\ell')} \frac{\phi_{v, t}}{\sum_{t'\in \Gamma^+(\ell)} \phi_{v, t'}} = \sum_{t\in \Gamma^{-}(u,\ell')} {\phi_{v, t}}.
        \end{equation*}
        Now assume the statement is true for the pair $(o, w)$ where $o \in \{v\} \cup I^+_v$ and $w \in I^+_v$ is a child of $o$.  Let $u$ be a child of $w$. Let us prove the statement for the pair $(w, u)$. Let $\ell' \in L$.

    Then we can write by induction,
    \begin{align*}
        \mathbb P[\ell_u = \ell'] &= \sum_{\ell'' \in L} \mathbb P[\ell_w = \ell'']\cdot \sum_{t\in \Gamma^+(\ell'') \cap \Gamma^-(u, \ell')} \frac{\phi_{w, t}}{\sum_{t'\in \Gamma^+(\ell'')}\phi_{w, t'}}\\
         &= \sum_{\ell'' \in L} \left(\sum_{t'\in \Gamma^-(w, \ell'')}\phi_{o, t'}\right)\cdot \sum_{t\in \Gamma^+(\ell'') \cap \Gamma^-(u, \ell')} \frac{\phi_{w, t}}{\sum_{t'\in \Gamma^+(\ell'')}\phi_{w, t'}}\\
         &= \sum_{\ell'' \in L} \left(\sum_{t'\in \Gamma^+(\ell'')}\phi_{w, t'}\right)\cdot \sum_{t\in \Gamma^+(\ell'') \cap \Gamma^-(u, \ell')} \frac{\phi_{w, t}}{\sum_{t'\in \Gamma^+(\ell'')}\phi_{w, t'}}\\
         &= \sum_{\ell'' \in L} \sum_{t\in \Gamma^+(\ell'') \cap \Gamma^-(u, \ell')} \phi_{w, t} \\
         &= \sum_{t\in \Gamma^-(u, \ell')} \phi_{w, t} \ . \qedhere
    \end{align*}
    \end{proof}

	Applying the lemma for $u \in \Gamma^+_v$, the parent $w$ of $u$, and a label $\ell' \in L$, we have $\mathbb P[\ell_u = \ell'] = \sum_{t\in \Gamma^-(u, \ell')}\phi_{w, t} = \chi_{u, \ell'}$.

        \subsection{Explicit linear program for $\LP^{(r, \ell^\circ)}$}
    \label{sec:explicit_LP}Instead of the recursive definition, we can now write an explicit one for $\LP^{(r, \ell^\circ)}$, where solvability and the size become clear. Towards this, define for each $k\in\{0,\dotsc,1/\epsilon\}$,
    \begin{align*}
        P_k = \{(v_0,\ell^\circ,v_1,\ell_1,\dotsc,v_k,\ell_k) :
        &\ v_0\in V^{(0)}, v_1 \in \Lambda^+(v_0), v_2\in \Lambda^+(v_1),\dotsc, v_k\in \Lambda^+(v_{k-1}) \\
        &\text{ and } \ell^\circ,\ell_1,\dotsc,\ell_k\in L\} \ .
    \end{align*}
    and $P = \bigcup_k P_k$.
    Each $p\in P$ corresponds to the vertices in super-layers on a path starting in the root and a choice of labels for these vertices.
    We will have two variables for each $p\in P$:
    $\chi_p$ and $x^{(p)}$.
    The variable $\chi_p$ describes whether the output labeling is consistent
    with $p$, i.e., all vertices along $p$
    are labeled with the label right after it.
    The variable $x^{(p)}$ describes $\chi_p$ times the solution vector contributed by the end vertex of $p$. 
    For each $p \in P_{k}$, $k< 1/\epsilon$, denote by $P^+(p)\subseteq P_{k+1}$ all paths with the prefix $p$, which extend
    the path by one vertex and one label.
    Consider the following linear program: 
        \begin{align}
        \chi_{(r,\ell^\circ)} &= 1 & & \label{LPC-compact:start} \\
        x^{(p)} &= \chi_p \cdot x^{(\ell)} & & \forall p = (\dotsc,\ell)\in P_{1/\epsilon}, \ell\in L \label{LPC-compact:xp-leaf} \\
        A x^{(p)} &\le \chi_p\cdot {\bf1} & & \forall p\in P \label{LPC-compact:packing}\\
	x^{(p)} &= \sum_{q\in P^+(p)} x^{(q)} & & \forall p\in P_0\cup \cdots \cup P_{1/\epsilon-1} \label{LPC-compact:xp} \\
        (\chi_q)_{q\in P^+(p)} &\in \chi_p \cdot \conv({\cal L}^+_{v,\ell}) & & \forall p=(\dotsc, v,\ell)\in P_0\cup\dotsc\cup P_{1/\epsilon-1} \label{LPC-compact:end}
	\end{align}

    In \eqref{LPC-compact:end}, the constraint involving ${\chi_p}\cdot\conv({\cal L}^+_{v,\ell})$ is implemented
    using the extended formulation from before. More precisely, we
    introduce the auxiliary variables and constraints (separately, for each $p$) and multiply the constant terms in the LP formulation with ${\chi_p}$.
    $(\chi_q)_{q\in P^+(p)}$ refers to
    the $(|n^\epsilon| \cdot |L|)$-dimensional vector with one entry for each $v,\ell$ where $(\dotsc,v,\ell) = q$.
    Since $|P| \le |L|^{O(1/\epsilon)}\cdot (n/\epsilon)$, it is clear that the above LP is solvable in time $|L|^{O(1/\epsilon)}\cdot \poly (n,m,d)$.

    \begin{lemma}
        Let {$(\chi_p,x^{(p)})_{p \in P}$} be a solution for LP(\ref{LPC-compact:start}-\ref{LPC-compact:end}).
        In time $|L|^{O(1/\epsilon)}\cdot \poly (n,m,d)$ we can compute a solution for ${\LP}^{(r,\ell^\circ)}$. 
    \end{lemma}
    \begin{proof} 
    For convenience, we rewrite the definition of ${\LP}^{(v,\ell)}$ here.
    \begin{align}
    Ax &\leq {\bf 1} \label{LPC-recursive:packing2}\\
    x &= \sum_{u \in \Lambda^+_v, \ell'} \chi_{u, \ell'} \cdot x^{(u, \ell')} \label{LPC-recursive:x2}\\
    \left(\chi_{u, \ell'}\right)_{u \in \Lambda^+(v), \ell' \in L} &\in \conv({\cal L}^+_{v, \ell}) \label{LPC-recursive:convex2}\\
    x^{(u, \ell')} &\in \text{LP}^{(u, \ell')} \text{ or } \chi_{u,\ell'} = 0 &\qquad &\forall u \in \Lambda^+_v, \ell' \in L \label{LPC-recursive:recurse2}
\end{align}

We prove by induction the following statement: 
For $p = (\dotsc,v,\ell) \in P_{1/\epsilon}$ with $\chi_p > 0$, we have $x^{(p)}/\chi_p \in \LP^{(v, \ell)}$. This holds for the case $\ell \in V^{(1/\epsilon)}$ as $x^{(p)}/\chi_p = x^{(\ell)} \in \LP^{(v, \ell)}$ by \eqref{LPC-compact:xp-leaf}.  Now assume $v \in V^{(\leq 1/\epsilon - 1)}$ and the statement holds for every child path of $p$. Consider the following solution to $\LP^{(v, \ell)}$:
\begin{align*}
    x := \frac{x^{(p)}}{\chi_p}, \quad x^{(u, \ell')} := \begin{cases}\frac{x^{((p, u, \ell'))}}{\chi_{(p, u, \ell')}} &\text{ if } \chi_{(p, u, \ell')} > 0 \\ 0 &\text{ otherwise} \end{cases},\chi_{u, \ell'} := \frac{\chi_{(p, u, \ell')}}{\chi_p}, \forall u \in \Lambda^+_v, \ell'  \in L.
\end{align*}
\eqref{LPC-recursive:packing2}, \eqref{LPC-recursive:x2} and \eqref{LPC-recursive:convex2} are satisfied; they respectively follow from \eqref{LPC-compact:packing}, \eqref{LPC-compact:xp} and \eqref{LPC-compact:end}. We have $x^{(u, \ell')} = \frac{x^{(q)}}{\chi_q}$ for the child path $q = (p, u, \ell')$ of $p$. Either $\chi_{u,\ell'} = 0$ or, by our induction hypothesis, we have $x^{(u, \ell')} \in \LP^{(u, \ell')}$. Therefore, we proved $\frac{x^{(p)}}{\chi_p} \in \LP^{(v, \ell)}$ and the statement holds for $p$. Thus the lemma follows. 

    \end{proof}

    To summarize, LP~(\ref{LPC-compact:start}-\ref{LPC-compact:end}) is a relaxation of the original problem
    and has size $|L|^{O(1/\epsilon)}\cdot \poly (n,m,d)$. After solving it while in time
    polynomial in its size, we derive a solution to $\LP^{(r,\ell^\circ)}$ with the same solution vector.

\newcommand{\bfe}{{\mathbf{e}}}
\subsection{Rounding algorithm}
    The technique to round the linear program is very similar to that in \cite{guo2022approximating}. We include the analysis for the sake of completeness.
    
    Suppose we have a solution for $\LP^{(r,\ell^\circ)}$.  We will recursively round the solution from $r$ to the leaves.  
    We let $\tilde \ell_r = \ell^\circ$, and call the recursive procedure rounding-without-cost$(r, \tilde \ell_r, x)$ for the root, where $x$ is the solution in $\LP^{(r, \ell^\circ)}$.  The final labeling vector will be $\tilde \ellv$.

    \begin{algorithm}[H]
    	\caption{rounding-without-cost$(v, \ell, x)$ 	    	    	\Comment{$x \in \LP^{(v, \ell)}$}}
	    \begin{algorithmic}[1]
	    	\If{$v \in V^{(1/\epsilon)}$} \Return \EndIf
	    	\State let $\big(x, (x^{(u, \ell')}, \chi_{u, \ell'})_{u \in \Lambda^+_v, \ell' \in L}\big)$ be the solution that certifies $x \in \LP^{(v, \ell)}$
	    	\State use Lemma~\ref{lem:rounding_pi} to randomly choose $\ellv \in {\cal L}^+_{v, \ell_v}$ so that $\Pr[\ell_u = \ell'] = \chi_{u, \ell'}$ for all $u \in \Lambda^+_v, \ell' \in L$
	    	\For{every $u \in \Lambda^+_v$}
	    		\State $\tilde \ell_u \gets \ell_u$
                        \State rounding-without-cost$(u, \ell_u, x^{(u, \ell_u)})$
	    	\EndFor
	    \end{algorithmic}
    \end{algorithm}
    
    Let $t  = \ln(1 + \epsilon/2)$. Let $\alpha_{1/\epsilon} = \bfe^{t}$. For every $i = 1/\epsilon - 1, 1/\epsilon-2, \cdots, 0$, we define $\alpha_i = \bfe^{\alpha_{i+1}-1} > \alpha_{i+1}$. 
    For a row $a$ of $A$, and a vertex $v \in V^{(\leq 1/\epsilon)}$, we define
    \begin{align*}
		\pack_{a, v} = a \cdot \sum_{\text{leaf descendant $u$ of $v$ in $T^+$}}  x^{(\tilde \ell_u)}.    
    \end{align*}
    Until the end of this section, we fix a row $a \in [0, 1]^\dimx$, and view $a$ as a row vector. 
	\begin{lemma}
		For every $v \in V^{(i)}$, consider the process rounding-without-cost$(v, \ell, x)$. We have 
		\begin{align*}
			\E\left[\exp\Big(\frac{t \cdot \pack_{a, v}}{n^\epsilon}\Big)\right] \leq \alpha_i^{ax/n^\epsilon}.
		\end{align*}
	\end{lemma}

	\begin{proof}
		We prove the lemma by inductions over $v$ from the bottom to the top of the collapsed tree.  When $v$ is a leaf,  $\pack_{a,v} = a x$ always holds. So, the left side of the inequality is $\exp\left(\frac{t\cdot ax}{n^\epsilon}\right) = \alpha_i^{ax/n^\epsilon}$,  by the definition of $\alpha_{1/\epsilon}$.
		
		We then assume $i < \frac1\epsilon$. Let $x^{(u, \ell')}, \chi_{u, \ell'}$ for every $u \in \Lambda^+_v, \ell' \in L$ be as obtained from Step 2 of the rounding algorithm. Then
		\begin{align}
			 \E\left[\exp\Big(\frac{t \cdot \pack_{a, v}}{n^\epsilon}\Big)\right] &= \E_{\ellv} \prod_{u \in \Lambda^+(v)}\E\left[\exp\Big(\frac{t \cdot \pack_{a, u}}{n^\epsilon}\Big) \mid \ell_u\right] \\
			&\leq \E_{\ellv}  \prod_{u \in \Lambda^+(v)} \alpha_{i+1}^{ax^{(u, \ell_u)}/n^\epsilon} \label{equ:induction} \\
			&= \E_{\ellv} \left[\alpha_{i+1}^{a(x|\ellv)/n^\epsilon}\right] \label{equ:define-x|l} \\
			&\leq \exp\left( \E_{\ellv}\big[\alpha_{i+1}^{a(x|\ellv)/n^\epsilon}\big] - 1\right) \\
			&\leq \exp\left((\alpha_{i+1}-1) ax/n^\epsilon\right) = \alpha_i^{ax/n^\epsilon}. \label{inequ:exponential-concavity}
		\end{align}
		\eqref{equ:induction} is by the induction hypothesis. In \eqref{equ:define-x|l}, we define $(x|\ellv) := \sum_{u \in \Lambda^+_v} x^{(u, \ell_u)}$. Therefore 
		\begin{align*}
			\E_{\ellv} (x|\ellv) = \E_{\ellv} \sum_{u \in \Lambda^+_v}x^{(u, \ell_u)} = \sum_{u \in \Lambda^+} \E_{\ellv} x^{(u, \ell_u)} = \sum_{u \in \Lambda^+, \ell' \in L} \chi_{u, \ell'}\cdot x^{(u, \ell')} = x.
		\end{align*}
		The inequality in \eqref{inequ:exponential-concavity} is by that $\E_{\ellv}[a(x|\ellv)] =ax$, $0 \leq a(x|\ellv) \leq n^\epsilon$ for every $\ellv$, and $\alpha_{i+1}^{\theta} - 1$ is a concave function of $\theta$. The equality is by the definition of $\alpha_i$. 
	\end{proof}

	\begin{lemma}
		\label{lemma:bound-alpha}
		For every $i \in[0,1/\epsilon]$, we have $\alpha_i \leq 1 + \frac{1}{i+1/\epsilon}$.
	\end{lemma} 
	\begin{proof}
		By definition, $\alpha_{1/\epsilon} = \bfe^t = 1+ \epsilon/2 = 1 + \frac{1}{1/\epsilon + 1/\epsilon}$ and thus the statement holds for $i = 1/\epsilon$.  Let $i \in [0, 1/\epsilon-1]$ and assume the statement holds for $i+1$. Then, we have 
		\begin{align*}
		\alpha_i &= \bfe^{\alpha_{i+1}-1} \leq \bfe^{\frac{1}{i+1+1/\epsilon}} \leq 1 + \frac{1}{i+1+1/\epsilon} + \left(\frac{1}{i+1+1/\epsilon}\right)^2\\
		&= 1 + \frac{i+2+1/\epsilon}{(i+1+1/\epsilon)^2} \leq 1 + \frac{1}{i + 1/\epsilon}.
		\end{align*}
		The first inequality used the induction hypothesis and the second one used that for every $\theta \in [0, 1]$, we have $e^\theta \leq 1 + \theta + \theta^2$. 
	\end{proof}

	Therefore, for the root vector $x \in \LP^{(r, \ell^\circ)}$, we have 
	\begin{align*}
		\E\left[\exp\left(\frac{t\cdot \pack_{a,r}}{n^\epsilon}\right)\right] \leq \alpha_0^{ax/n^\epsilon} \leq (1 + \epsilon)^{ax/n^\epsilon} \leq 1+\epsilon.
	\end{align*}
	
	Therefore for every $\delta > 0$, we have
	\begin{align*}
		\Pr\left[\pack_{a, r} \geq \delta n^\epsilon \right] \leq \frac{1+\epsilon}{\exp(t\delta)} =(1+\epsilon)\cdot(1+\epsilon/2)^{-\delta}.
	\end{align*}
	
	Setting $\delta = O\left(\frac{\log \nconstraints}{\epsilon}\right)$ and applying the union bound over all the $\nconstraints$ constraints, we can guarantee that with probability at least $1/2$, all the $\nconstraints$ constraints are satisfied. Repeating the process $O(\log (nm))$ times can increase the success probability to $1-\poly(n, m)$. This proves Theorem~\ref{thm:tree_labeling_approx}. 

    Recall that in the Perfect Binary Tree Labeling instance reduced from Additive-DP, we have $n = \poly(\sizebound)$ and $|L| = \poly(\sizebound, \listslength)$. This implies Theorem~\ref{thm:main} holds for the case $c = {\bf0}$.
\section{Cost-preserving rounding for Perfect Binary Tree Labeling}
\label{sec:costs}
We refer to Section~\ref{sec:perfect-TL} for the definition of Perfect Binary
Tree Labeling. Here, we present a more involved rounding
algorithm that ensures that the costs of the LP solution do not increase
by the rounding. This is done by comining the ideas from the previous
rounding with techniques from~\cite{rohwedder2025cost}.
Formally, in this section we prove the following theorem.
\begin{theorem}
\label{thm:tree_labeling_costs}
    For every $\epsilon > 0$, there is a randomized multi-criteria approximation algorithm for Perfect Binary Tree Labeling
    that in time $|L|^{O(1/\epsilon)}\cdot \poly(n, \nconstraints, \dimx)$ finds 
    a labeling with corresponding solution vector $x^{(\ellv)}$ such that:
    \begin{itemize}
        \item $c\T x^{(\ellv)} \le \opt$ and
        \item $Ax^{(\ellv)} \le O\left(\frac{n^\epsilon}{\epsilon^2} \log \nconstraints\right) \cdot {\bf1}$,
    \end{itemize}
    assuming there exists a labeling $\ellv^*$ with solution vector $x^{(\ellv^*)}$ such that $c\T x^{(\ellv^*)} \le \opt$ and
     $Ax^{(\ellv^*)} \le {\bf 1}$.
\end{theorem}
Recall that in Section~\ref{sec:perfect-TL} we
gave a linear program of size $|L|^{O(1/\epsilon)}\cdot \poly (n,m,d)$ over
variables $x^{(r,\ell^\circ)}$ that capture the solution vector of the continuous relaxation. For the purpose of this section, we solve the same linear program, but
minimize $c\T x^{(r,\ell^\circ)}$ over it.

\subsection{Polynomial-size decomposition of  $\chi \in \conv({\cal L}^+_{v, \ell})$} 
    
    The naive rounding algorithm in Section~\ref{sec:extended_conv} implicitly gives a decomposition of any feasible solution $(\phi, \chi)$ to the LP(\ref{LPC-convex:start}-\ref{LPC-convex:end}) into a convex combination of integral feasible solutions. However, the support of the convex combination may have exponential size. In our algorithm with costs, it is crucial for us to construct a polynomial-size convex combination. The goal of this section is to prove the following lemma:
    \begin{lemma}
    \label{lem:poly_size_decomp}
    	Given a $\chi \in \conv({\cal L}^+_{v, \ell})$, we can output in polynomial time a decomposition of $\chi$ into a convex combination of extreme points in $\conv({\cal L}^+_{v, \ell})$. 
    \end{lemma}
\begin{proof}
	Given $\chi \in \conv({\cal L}^+_{v, \ell})$, we consider the following exponential-sized LP which tries to find the decomposition explicitly:
        \begin{align*}
            \sum_{\ellv\in {\cal L}^+_{v,\ell}} \lambda_{\ellv} &= 1 \\
            \sum_{\ellv\in {\cal L}^+_{v,\ell} : \ellv_u = \ell'} \lambda_{\ellv} &= \chi_{u, \ell'} &\quad &\forall u \in \Gamma^+_v, \ell' \in L \\
            \lambda_{\ellv} &\ge 0 &\quad &\forall \ellv \in {\cal L}^+_{v, \ell}
        \end{align*}
        The dual of this linear program is
        \begin{align*}
            \min \quad \pi + &\sum_{u \in \Lambda^+_v, \ell' \in L} \chi_{u, \ell'} \cdot \tau_{u,\ell'} & & \\
            \pi + \sum_{u \in \Lambda^+_v} \tau_{u,\ell_u} &\ge 0 & \quad & \forall \ellv\in {\cal L}^+_{v,\ell}
        \end{align*}
        The dual is clearly feasible, since zero is a solution.
        Thus, it is bounded if and only if the primal is feasible.
        We run the Ellipsoid method with a separation oracle on the
        dual obtaining a polynomial number of constraints
        that prove its boundedness. These constraints (or 
        the corresponding label tuples) suffice as variables
        of the primal to obtain a feasible solution.
        We therefore solve the primal restricted to these variables
        and obtain a solution $\lambda$.
        It remains to show that the separation problem of the
        dual can be solved. This is to find a label tuple $\ellv$
        such that $\sum_{{u\in \Lambda^{+}_v}} \tau_{u,\ell_u} < -\pi$.
        This can be reformulated as optimizing a linear objective
        over $\conv({\cal L}^+_{v,\ell})$, which we can do
        by linear programming, since we have a polynomial-size
        extended formulation of $\conv({\cal L}^+_{v,\ell})$.
%
%
\end{proof}    

\subsection{A useful tool: Semi-Random Rounding}
\label{sec:semi-random-rounding}
We describe  a semi-random rounding procedure, described in Algorithm~\ref{alg:semi-random-matching}. The notation in this section is independent of the other sections. 
\begin{algorithm}[H]
	\caption{Semi-Random Rounding}
	\label{alg:semi-random-matching}
    \begin{algorithmic}[1]
        \Require{a finite domain $D$, a partition $\cal P$ of $D$, $\lambda \in [0, 1]^D$ such that $\lambda(S) \in \Z, \forall S \in \cal P$, and $K \in \Z_{>0}$, and $c \in \R^D$
        }
        \Ensure{a random vector $\lambda' \in \{0, 1\}^D$}
        	\State \label{step:semi-random-find-lambdaK} let $\lambda^{(K)}$ be the vector with the minimum $c\T \lambda^{(K)}$ satisfying $\lambda^{(K)}_i \in \Big\{\frac{1}{2^K}\cdot \floor{2^K \lambda_i}, \frac{1}{2^K}\cdot \ceil{2^K \lambda_i} \Big\}, \forall i \in D$ and $\lambda^{(K)}(S) = \lambda(S)$ for every $S \in \cal P$ 
        \For{$k = K-1$ down to $0$}
            \State choose an arbitrary perfect matching $M$ between the indices $S:= \{i \in D: 2^k \lambda^{(k+1)}_i \notin \Z\}$, such that $(i, j)\in M$ implies that $i$ and $j$ are in the same partition in $\cal P$
            \State let $\nu_i = 0$ for every $i \in D \setminus S$
            \For{every $(i, j) \in M$} 
                \State with probability $1/2$, let $\nu_i \gets \frac12, \nu_j \gets -\frac12$; with probability $1/2$, let $\nu_i \gets -\frac12, \nu_j \gets \frac12$
            \EndFor
            \State 
            choose $\nu' \in \{\nu, -\nu\}$ so that $c\T \nu' \leq 0$
            \State $\lambda^{(k)} \gets \lambda^{(k+1)} + \frac{\nu'}{2^k}$
        \EndFor
        \State \Return  $\lambda' := \lambda^{(0)}$
    \end{algorithmic}
\end{algorithm}

\begin{claim}\label{cla:lambda}
	The output $\lambda'$ of Algorithm~\ref{alg:semi-random-matching} satisfies $\lambda'(S) = \lambda(S)$ for  every $S \in \cal P$. Moreover, $c\T \lambda' \leq c\T \lambda$ always holds. 
\end{claim}
\begin{proof}
    First, we show that $c\T\lambda^{(K)}\le c\T\lambda$. This holds as $\lambda$ is in the convex hull of all vectors $\lambda^{(K)}$ satisfying the condition in Step \ref{step:semi-random-find-lambdaK} of Algorithm~\ref{alg:semi-random-matching}. So, the $\lambda^{(K)}$ with the minimum $c\T\lambda^{(K)}$ clearly has $c\T\lambda^{(K)}\le c\T\lambda$. 
    Algorithm~\ref{alg:semi-random-matching} ensures that $\lambda^{(k)}=\lambda^{(k+1)}+\frac{\nu'}{2^k}$ with $c^T\nu' \le 0$ at every iteration. Hence $c\T\lambda'\le c\T\lambda^{(K)} \leq c\T \lambda$. 
\end{proof}

\begin{lemma}
    \label{lemma:semi-random-rounding}
	Consider Algorithm~\ref{alg:semi-random-matching}. Let $c' \in [0, 1]^D$ be any vector.  Let $\delta, \beta \in (0, 0.1)$. Then we have
    \begin{align*}
        \Pr\left[{c'}\T \lambda' \leq (1+7\delta) \max\left\{{c'}\T \lambda + \frac{\|c'\|_1}{2^K}, \frac{3}{\delta^2}\ln\frac1\beta\right\}\right] \geq 1 - 4\beta.
    \end{align*}
\end{lemma}
\begin{proof}
	Clearly, we have $|{c'}\T \lambda^{(K)} - {c'}\T\lambda| < \frac{\|c'\|_1}{2^K}$. 
	
    Let us now estimate the increase ${c'}\T \lambda^{(k)} - {c'}\T \lambda^{(k+1)}$ induced by the rounding in iteration $k$. Let $M, S$, $\nu$ and $\nu'$ be as in the iteration. For each $p = (i, j) \in M$, let 
    \begin{align*}
        X_p = c'_i \left(\frac12 + \nu_i\right) + c'_j \left(\frac{1}{2} + \nu_j\right) \in \{c'_i, c'_j\} \in [0, 1].
    \end{align*} 
    We have $\E[X_p] = \frac{1}{2}(c'_i + c'_j)$.  Therefore, 
    \begin{align*}
        \mu:=\E\left[\sum_{p \in M}X_p\right] = \frac{1}{2}\sum_{i \in S}c'_i \leq 2^k \sum_{i \in S}c'_i \lambda^{(k+1)}_i \leq 2^k \cdot {c'}\T \lambda^{(k+1)}.
    \end{align*}
    The first inequality used that every $i \in S$ has $\lambda^{(k+1)}_i \geq \frac{1}{2^{k+1}}$.
    
By a variant of Chernoff bound, we have
\begin{align*}
    &\quad \Pr\left[\sum_{p \in M} X_p > \mu + 2^{-k/3}\delta \cdot \max\left\{2^k \cdot {c'}\T \lambda^{(k+1)}, \frac{3\cdot 2^k}{\delta^2}\ln \frac1\beta\right\}\right]\\
    &\leq \exp\left(-\frac{(2^{-k/3}\delta)^2}{3} \cdot \frac{3\cdot 2^k}{\delta^2}\ln \frac1\beta\right) = \exp\left(-2^{k/3}\ln \frac{1}{\beta}\right)  = \beta^{2^{k/3}}.
\end{align*}
Notice that ${c'}\T \nu = \sum_{i \in S} c'_i \nu_i = \sum_{p \in M} X_p - \mu$. We have
\begin{align*}
   	\Pr\left[\frac{1}{2^k}\cdot {c'}\T \nu > 2^{-k/3}\delta \cdot \max\left\{{c'}\T \lambda^{(k+1)}, \frac{3}{\delta^2}\ln \frac1\beta\right\}\right] \leq \beta^{2^{k/3}}.
\end{align*}
Using the same argument, the above inequality holds with $\nu$ replaced by $-\nu$. Therefore, 
\begin{align*}
    \Pr\left[\frac{1}{2^k}\cdot {c'}\T \nu' > 2^{-k/3}\delta \cdot \max\left\{{c'}\T \lambda^{(k+1)}, \frac{3}{\delta^2}\ln \frac1\beta\right\}\right] \leq 2\beta^{2^{k/3}}.
\end{align*}
We assume for every $k = 0, 1, 2, \cdots, K-1$, we have 
\begin{align*}
	\frac{1}{2^k}\cdot {c'}\T \nu' \leq 2^{-k/3}\delta \cdot \max\left\{ {c'}\T \lambda^{(k+1)}, \frac{3}{\delta^2}\ln \frac1\beta\right\}.
\end{align*}
This happens with probability at least $1 - \sum_{k = 0}^{K-1}2\beta^{2^{k/3}} \geq 1 - 4\beta$ as $\beta < 0.1$.  Under this condition, for every $k$, we have 
\begin{align*}
    {c'}\T \lambda^{(k)} = {c'}\T \lambda^{(k+1)} + \frac1{2^k}{c'}\T \nu' \leq \left(1+2^{-k/3}\delta\right) \cdot \max\left\{{c'}\T \lambda^{(k+1)}, \frac{3}{\delta^2}\ln \frac1\beta \right\}.
\end{align*}
Consider the smallest $k' \in \{0, 1, 2, \cdots, K-1\}$ for which we have ${c'}\T \lambda^{(k'+1)} < \frac{3}{\delta^2}\ln \frac1\beta$. If $k'$ does not exist, we have
\begin{align*}
    {c'}\T \lambda^{(0)} \leq  {c'}\T \lambda^{(K)} \cdot \prod_{k = 0}^{K-1} \left(1 + 2^{-k/3}\delta\right) \leq \exp\left(\frac{\delta}{1-2^{-1/3}}\right) \cdot {c'}\T \lambda^{(K)} \leq (1 + 7\delta)\cdot {c'}\T \lambda^{(K)}. 
\end{align*}
If $k'$ exists, then
\begin{align*}
    {c'}\T \lambda^{(0)} &\leq \frac{3}{\delta^2}\ln \frac1\beta \cdot \prod_{k=0}^{k'}(1 + 2^{-k/3}\delta) \leq (1+7\delta)\cdot \frac{3}{\delta^2}\ln \frac1\beta.
\end{align*}
The lemma follows from that $\lambda' = \lambda^{(0)}$ and ${c'}\T \lambda^{(K)} < {c'}\T \lambda + \frac{\|c'\|_1}{2^K}$.
\end{proof}

    \subsection{Rounding algorithm}
%
    With the tool of semi-random rounding, we can now design our rounding algorithm when we have $c \neq 0$. 
	We shall use the vector $\tilde \ell:= (\tilde \ell_v)_{v \in V^{(\leq 1/\epsilon)}}$ to store the final labels for all vertices in the collapsed tree.  Suppose we have a solution $x \in \LP^{(r,\ell^\circ)}$ minimizing $c\T x$.  We shall define $\tilde \ell_r = \ell^\circ$.  We will recursively round the solution from $V^{(i)}$ for $i$ from $0$ to $1/\epsilon - 1$. Focus on some $i$. We already obtained a label $\tilde \ell_v$ for every $v \in V^{(i)}$, along with a feasible solution $\big(x, (x^{(u, \ell')}, \chi_{u, \ell'})_{u \in \Lambda^+_v, \ell' \in L}\big)$  to ${\LP}^{(v,\tilde \ell_v)}$ for every $v \in V^{(i)}$. \notesl{Check the style conventions.}
    
    For every $v \in V^{(i)}$, we use {Lemma \ref{lem:poly_size_decomp}} to find a decomposition of $(\chi_{u, \ell'})_{u \in \Lambda^+_v, \ell' \in L}$ into a convex combination $(\lambda_{\ellv})_{\ellv \in {\cal L}^+_{v, \ell_v}}$ (of polynomial support) of valid label vectors in ${\cal L}^+_{v, \tilde \ell_v}$.  We then define $D$ to be the set of all label vectors $\ellv \in {\cal L}^+_{v, \tilde \ell_v}$ in the decomposition for $v$ over all $v \in V^{(i)}$. We define a partition $\cal P$ of $D$ according to $v$: for every $v \in V^{(i)}$, the set of all label vectors $\ellv \in {\cal L}^+_{v, \tilde \ell_v}$ in the decomposition for $v$ forms a set in the partition $\cal P$.  We assume each $\ellv \in D$ has its own identity. 
    
    For every $\ellv \in D$ that is in the decomposition for $v$, we define $\tilde c_{\ellv} = c\T  \sum_{u \in \Lambda^+_v} x^{(u, \ell_u)}$.
    We then run Algorithm~\ref{alg:semi-random-matching} for this $D$, $\lambda \in [0, 1]^D$, $\cal P$, $K$ to be decided later, and the vector $\tilde c$.  This gives us an integral vector $(\lambda'_{\ellv})_{\ellv \in D}$, which naturally gives the labels for $V^{(i+1)}$ as follows. For every $v \in V^{(i)}$, consider the unique $\ellv \in {\cal L}^+_{v, \tilde \ell_v}$ with $\lambda'_{\ellv} = 1$. Then we define $\tilde \ell_u = \ell_u$ for every $u \in \Lambda^+_v$. We can then recurse to the next layer by selecting ${x^{(u,\tilde \ell_u)}}$ as feasible solution to ${\LP}^{(u,\tilde \ell_u)}$.

    We need to show two things while applying this rounding to move to the next layer. First, we show that the cost does not increase, second we will prove that the packing constraints do not increase much, with high probability.

    For the cost, we define the current cost as follows
    \begin{equation*}
        \cost^{(i)} = \sum_{v\in V^{(i)}}\sum_{\ellv\in {\cal L}^+_{v,\tilde \ell_v}} \tilde c_{\ellv}\cdot \lambda_{\ellv}\ .
    \end{equation*}
    Note that we can write 
    \begin{align*}
        \cost^{(0)} &= \sum_{v\in V^{(0)}}\sum_{\ellv\in {\cal L}^+_{v,\tilde \ell_v}} \tilde c_{\ellv}\cdot \lambda_{\ellv} \\
        &= \sum_{v\in V^{(0)}}\sum_{\ellv\in {\cal L}^+_{v,\tilde \ell_v}} \lambda_{\ellv}\cdot  c\T  \sum_{u \in \Lambda^+_v}x^{(u, \ell_u)}\\
        &= \sum_{v\in V^{(0)}} c\T  \sum_{u \in \Lambda^+_v, \ell'} \chi_{u,\ell'}\cdot x^{(u, \ell')}\\
        &= c\T x\ .
    \end{align*}
    So $\cost^{(0)}$ is indeed equal to the initial fractional cost of our LP solution. By the same arguments, one can see that $\cost^{(1/\epsilon)}$ is equal to the final cost of our integral solution.
    \begin{lemma}
     Let $i\in \{0,1,\dotsc,1/\epsilon - 1\}$.
        It holds that $\cost^{(i+1)} \le \cost^{(i)}$ with probability $1$.
    \end{lemma}
    \begin{proof}
    	This follows easily from Claim~\ref{cla:lambda}.
    \end{proof}
    Applying the above lemma on all layers, we obtain that the cost of the final integral solution is at most the cost of the initial fractional solution.
    
 \subsection{Analysis of the packing constraints} 

 For the packing constraints, we will proceed similarly as for the cost, with the difference that we will use Lemma~\ref{lemma:semi-random-rounding} in place of Claim~\ref{cla:lambda}. Focus on some level $i \in \{0, 1, \dotsc, 1/\epsilon - 1\}$. For each row $a$ of the matrix $A$, and $\ellv \in {\cal L}^+_{v, \tilde \ell_v}$ for some $v \in V^{(i)}$,  we define 
 \begin{equation*}
     c(a)_{\ellv} = a  \cdot \sum_{u \in \Lambda^+_v} x^{(u, \ell_u)}\ .
 \end{equation*}
 We assume for each $\ellv$, we know the $v \in V^{(i)}$ with $\ellv \in {\cal L}^+_{v, \tilde \ell_v}$.
 
 Note that 
 \begin{equation*}
     c(a)_{\ellv} = \sum_{u \in \Lambda^+_v}a \cdot x^{(u, \ellv_u)} \le |\Lambda^+_v|=n^\epsilon\ ,
 \end{equation*}
 so we will apply Lemma~\ref{lemma:semi-random-rounding} with $c'=c(a)_{\ellv}/n^\epsilon$ to bound the violation of the packing constraints\footnote{This is the key moment where the factor $n^\epsilon$ will appear in our approximation guarantee.}. Moreover, we define
 \begin{equation*}
     \pack_a^{(i)}:=\sum_{v\in V^{(i)}}\sum_{\ellv\in {\cal L}^+_{v,\tilde \ell_v}} c(a)_{\ellv}\cdot \lambda_{\ellv}\ .
 \end{equation*}
 Note that we can write, as in the case of the cost vector,
\begin{align*}
        \pack_a^{(0)} &= \sum_{v\in V^{(0)}}\sum_{\ellv\in {\cal L}^+_{v,\tilde \ell_v}} c(a)_{\ellv}\cdot \lambda_{\ellv} \\
        &= \sum_{v\in V^{(0)}}\sum_{\ellv\in {\cal L}^+_{v,\tilde \ell_v}} \lambda_{\ellv}\cdot  a  \sum_{u \in \Lambda^+_v} x^{(u, \ell_u)}\\
        &= \sum_{v\in V^{(0)}} a  \sum_{u \in \Lambda^+_v, \ell'} \chi_{u,\ell'}\cdot x^{(u, \ell')}\\
        &= a\cdot  x\ .
    \end{align*}
Hence, we start with $\pack_a^{(0)}\le 1$, and the final value of that constraint is equal to $\pack_a^{(1/\epsilon)}$. In the following, we will apply our semi-random rounding. 
\begin{lemma}
    \label{lem:rounding}
    By selecting $K=O(\log n)$ in Algorithm~\ref{alg:semi-random-matching}, for any row $a$ of matrix $A$, and any $i\in \{0,1,\dotsc,1/\epsilon - 1\}$, it holds that 
    \begin{align*}
        \Pr\left[\frac{\pack_a^{(i+1)}}{n^\epsilon} \leq (1+\epsilon) \max\left\{\frac{\pack_a^{(i)}}{n^\epsilon} + 1, \frac{100}{\epsilon^2}\ln m\right\}\right] \geq 1 - 1 / m^2.
    \end{align*}
\end{lemma}
\begin{proof}
    Notice that $\|c(a)\|_1\le \poly(n)$ since there are only a polynomial number of non-zero $\lambda_{\ellv}$s. With this in mind, it is easy to apply Lemma~\ref{lemma:semi-random-rounding} with $c'=c(a)$, $\beta=1/(4m^2)$, $\delta = \epsilon/7$, $K = \log(\|c(a)\|_1)$.
\end{proof}
Using Lemma~\ref{lem:rounding}, we can conclude easily as follows. We repeat the semi-random rounding $O(1+\log(1/\epsilon)/\log(m))$ times to boost the probability of success of Lemma~\ref{lem:rounding} from $1-1/m^2$ to $1-\epsilon/m^2$. Now with probability at least $1/2$, the bound in Lemma~\ref{lem:rounding} holds for all rows of $A$, and for every iteration of the rounding. In that case, it is easy to prove by induction on $i$ that 
\begin{equation*}
   \frac{\pack_a^{(i)}}{n^\epsilon}\le (1+\epsilon)^i \cdot \frac{200 \ln m}{\epsilon^2}\ ,
\end{equation*}
for all $i\le 1/\epsilon$. We can repeat the whole procedure enough times to increase the probability to $1-1/\poly(n)$. This concludes the proof of Theorem~\ref{thm:tree_labeling_costs} and therefore
for Theorem~\ref{thm:main} for the general case with a cost vector.

\section{Applications with packing constraints}
\label{sec:app-packing}
Our framework applies most naturally, although not exclusively to problems involving packing constraints.  Throughout the paper, we shall use $\oplus$ to denote the sumset operator: for two sets $A$ and $B$ of vectors, we have $A \oplus B := \{a + b: a \in A, b \in B\}$.

\subsection{Robust Shortest Path}
For completeness we repeat the dynamic program for modeling paths that was given also in 
the introduction.  Without loss of generality, we assume the input graph $G = (V, E)$ is a DAG.   We have one subproblem $I(v)$ for each vertex $v\in V$, with $I^\circ := I(s)$ being the root problem. The solutions are vectors in $\{0, 1\}^E$, with $S(I(v)):=\{\mathbb{1}_P: P\text{ is a path from $v$ to $t$}\}$ for each $v \in V$.  The relation $\prec$ is defined by the topological order of the vertices in $G$.  The recursion for the solutions are defined as follows: For the subproblem $I(t)$, we have $S(I(t)) = \{\bf0\}$. For every vertex $v \neq t$, we have 
 \begin{align*}
     S(I(v)) = \bigcup_{(v, u) \in \delta^+(v)} S(I(u)) \oplus \{\mathbb{1}_{(v, u)}\}.
 \end{align*}

Suppose that vector $c\in \R_{\geq 0}^E$ describes the cost of each edge, then finding the shortest $s-t$ path is equivalent to finding $x\in S(I(s))$ minimizing $c\T x$. On top of the problem, we define packing constraints $(a^j)^\rmT x \leq 1$ for every $j \in [k]$, capturing the length requirements. 

Using Theorem~\ref{thm:main} we find a solution to $I(s)$ minimizing $\sum_{e\in E} c_e x_e$.  Our solution will be a $s$-$t$ path with cost at most $\opt$, that violates the length requirements by at most $O(\frac{n^\epsilon}{\epsilon^2}\log k)$.

\subsection{Generalized Flow}
    Consider a dynamic program of the form Additive-DP with subproblems $I(s', F', w)$ which intuitively correspond to Generalized Flow instances with source $s'\in V$ with excess $F'$ and a bound of $w$ on the $\ell_1$-norm of the flow.
    Both $F'$ and $w$ are bounded by the total flow that we have from the input.
    
    Formally, the solutions $S(I(s', F',w))$ are all $x \in \ZZ^E$ satisfying
    \begin{align*}
       \sum_{e\in \delta^+(s')} x_e - \sum_{e\in \delta^-(s')} g(e) x_e &= F' & \\
       \sum_{e\in \delta^+(v)} x_e - \sum_{e\in \delta^-(v)} g(e) x_e &= 0 & \forall v\in V\setminus \{s'\} \\
       \sum_{e\in E} x_e &\le w \ ,
    \end{align*}
    and where in the graph consisting of edges with non-zero flow, i.e., $x_e > 0$,
    every edge is reachable from $s'$. Note that we can always assume, in
    the optimal solution to Generalized Flow all edges carrying flow are reachable,
    since otherwise we can reduce the flow on these to zero, which maintains the requirements on the flow and only reduces its cost.
    Then we have
    \begin{equation*}
        S(I(s', 0, 0)) = \{ {\bf0} \} \text{ and }
        S(I(s', F', 0)) = \emptyset \text{ for } F' > 0 \ .
    \end{equation*}
    For $w\ge 1$, the solutions satisfy the recurrence
    \begin{align*}
        S(I(s',F', w)) &= S(I(s',F',w-1)) \cup \bigcup_{w'\in [w-1], e=(s',v)\in \delta^+(s')} \hspace{-3em} \{\mathbb{1}_e\} \oplus S(I(s', F'-1, w')) \oplus S(I(v, g(e),w - w' - 1)).
    \end{align*}
    
Above, $\mathbb{1}_e$ is the indicator vector for the edge $e$. We implement the capacities by the additional packing constraints $x_e \leq \capa_e$ for every $e \in E$.  Theorem~\ref{thm:gen-flow} follows by applying Theorem~\ref{thm:main} to the subproblems above with the given packing constraints and minimizing the total cost of the generalized flow.
    
\subsection{Longest Common Subsequence}
    We encode a common subsequence using vectors $x\in \{0,1\}^{m\cdot n}$. Thus, we have a dimension for every pair of character $(a[i],b[j])$. 
    Then we consider a dynamic program of the form Additive-DP with subproblems $I(i,j, k)$ which corresponds to the goal of finding a common subsequence of length $k$ between the two partial strings $a[1]a[2]\ldots a[i]$ and $b[1]b[2]\ldots b[j]$, for every $i \in [0, n], j \in [0, m]$ and $k \in [0, \opt]$. The base cases are when $i=0$ or $j=0$, in which case we have $S(I(i,j, 0))=\{\bf0\}$, and $S(I(i, j, k)) = \emptyset$ when $k \geq 1$. Otherwise, we obtain the following recurrence:
    \begin{equation*}
        S(I(i,j, k)) = 
        \begin{cases}
            \big( S(I(i-1, j-1, k-1)) \oplus \{\mathbb{1}_{ij}\}\big)
            \cup  S(I(i-1,j, k))\cup S(I(i,j-1, k)) &\text{ if } a[i]=b[j]\text{ and } k \geq 1\\
             S(I(i-1,j, k))\cup S(I(i,j-1, k)) &\text{ otherwise.}
        \end{cases}
    \end{equation*}
    To model the bounded number of repetitions for each character, we add packing constraints
    \begin{equation*}
        \sum_{i,j : a[i] = b[j] = z} \frac{1}{C} \cdot x_{i,j} \le 1 \qquad \forall z\in\Sigma \ .
    \end{equation*}
    There are no costs.  It is now clear that Theorem \ref{thm:main} applies.

\section{Applications with covering constraints}
\label{sec:app-covering}
Despite our framework being based on packing constraints, we can
also apply it on some problems involving covering constraints.
In this section, we demonstrate this on the Directed Steiner Tree Problem
and Colorful Orienteering.


\subsection{Robust Steiner Tree Cover}

Recall that we are given an $n$-vertex directed graph $G=(V,E)$, a root $r \in V$,  edge costs $c \in \R_{\geq 0}^E$, a bound $B \in \R_{\geq 0}$ on the cost, a set of $k$ terminals $K\subseteq V$, and vectors $a^1, a^2, \cdots, a^k \in [0, 1]^E$. Our goal is to find an out-arborescence $T$ rooted at $r$ with cost at most $B$, and $\sum_{e \in T}c_e \leq B$, $\sum_{e \in T}a^j_e \leq 1$ for every $j \in [k]$, that contains the maximum number of terminals in $K$. 

We introduce a subproblem $I(v, o, w)$ for each $v\in V$, $o\in \{0,1,\dotsc,\opt\}$ (we can assume $\opt$ is given by guessing) and $w \in \{0,1,\dotsc,n\}$. Intuitively, we want the solutions $S(I(v, o, w))$ to represent out-arborescences rooted at $v$ with $o$ terminals and at most $w$ edges. However, by
limitations of Additive-DP, we cannot avoid multi-edges or possibly multiple incoming edges to a single vertex from occuring.
To capture this, we associate with each $x\in \ZZ^E$ a weighted graph $G(x)$ on the vertices $V$,
which contains edge $e$ with weight $x_e$ for every $e\in E$ where $x_e > 0$.
Formally, the solution $x \in S(I(v, o, w))$ are all $x\in \ZZ^E$ where $G(x)$
has total weight at most $w$, all edges in $G(x)$ are reachable from $v$ (through paths in $G(x)$) and
where the total incoming weight to terminal vertices is $o$.
The recurrence comes from the following idea: for $x \in S(I(v, o, w))$ consider any edge $e = (v,u)$ in $G(x)$, i.e., an outgoing edge from $v$.
Every edge $e'$ in $G(x)$ is reachable from $u$ or from $v$ through a path that does not contain $e$ (except possibly as the last edge if $e' = e$).
We let $x'_{e''} = x_{e''}$ for every edge $e'' \neq e$ reachable from $u$ and $x'' = x - x' - \mathbb{1}_e$.
Then $x'\in S(I(u, o', w'))$ and $x''\in S(I(v, o-o',w-w'-1))$ for some $o'\le o$ and $w'\le w-1$.

In the formalism of Additive-DP, we write the solutions recursively as follows. For any edge $e = (v, u)$, we have $S(I(v, 1, 1)) = \{{\mathbb{1}_{e}}\}$ if $u\in K$ and $S(I(v, 0, 1)) = \{{\mathbb{1}_{e}}\}$ if $u\notin K$. Any other subproblem
with $w\le 1$ is infeasible. 
For any $I(v, o, w)$ with $w \geq 2$, we have
\begin{align*}
    S(I(v, o, w)) &:= S(I(v, o, w - 1)) \\
    &\quad \cup \bigcup_{(v, u) \in \delta^+(v), o' \in [0, o], w' \in [w-1]} S(I(v, o', w'))  \oplus S(I(u, o - o', w - w' - 1)) \oplus \{\mathbb{1}_{(v, u)}\}.
\end{align*}
The root problem is $I^\circ := I(r, \opt, n)$. 

Finally, we add packing constraints that enforces that each terminal should not have incoming weight more than $1$; i.e., $\sum_{e \in \delta^-(t)}x_e \leq 1$ for every $t \in K$.  We also add the constraints $(a^j)^\rmT x\leq 1$ for every $j \in [k]$.  Let $\alpha = O\left(\frac{n^\epsilon}{\epsilon^2} \cdot \log (k + |K|)\right)$ as in  Theorem~\ref{thm:main}.
Using the theorem, we find a solution $x$ with cost at most $B$, where each terminal has at most $\alpha$ incoming edges. In particular, at least $\opt / \alpha$ terminals are reachable from $r$.  For every $j \in [k]$, we have $(a^j)^\rmT x \leq \alpha$. This finishes the proof of Theorem~\ref{thm:DST}.

\subsection{Colorful Orienteering}
Recall in Colorful Orienteering that we are given a directed graph $G = (V, E)$ with edge costs $c \in \R_{\geq 0}^E$ and colors $\kappa \in [C]^E$,  a start $s\in V$, destination $t\in V$, and budget $B\in \R_{\geq 0}$. The goal is to find a not-necessarily-simple path $p$ from $s$ to $t$ of cost at most $B$, that covers the maximum number of colors. 
One can prove that an optimum path has at most $O(n^2)$ edges (counting multiplicities). This holds as we must visit a new color in every $n$ edges; otherwise, the path can be shortcut. Therefore, by creating an $O(n^2)$-level graph, we can assume the graph $G$ is a DAG, and any directed path is simple. The number of vertices becomes $O(n^3)$ and this will not affect our claimed running time and approximation ratio. 

It may be necessary that the optimum path visits the same color many times. For the Steiner Tree version, it was important
that the optimal solution only covers each terminal once. Hence, this approach needs
to be refined here, allowing colors to be not used even though they are visited by
a subsolution.

We have a problem $I(v, j)$ for every $v \in V$ and integer $j \in [0, \opt]$, where $\opt$ is the guessed optimum number of colors.  A solution to the problem $I(v, j)$ is vector $(x \in \{0, 1\}^E, y \in \Z_{\geq 0}^{C})$, where $x$ indicates a directed path from $v$ to $t$, $y_b$ for each $b \in [C]$ is \emph{at most} the number of edges in the path with color $b$, and $\sum_{b \in [C]}y_b = j$. 

We have $S(I(t, 0)) = \{{\bf0}, {\bf0}\}$. For any $v \neq t$ and $j \geq 1$, we have 
\begin{align*}
    S(I(v, j)) := \bigcup_{(v, u) \in \delta^+(v)} \Big[\big(S(I(u, j)) \oplus \{(\mathbb{1}_{(v, u)}, {\bf0})\}\big) \cup \big(S(I(u, j - 1)) \oplus \{(\mathbb{1}_{(v, u)}, \mathbb{1}_{\kappa(v, u)})\}\big)\Big].
\end{align*}

The root problem is $I(s, \opt)$. We have $I(v, j) \preceq I(v', j')$ if there is a path from $v'$ to $v$ in $G$, and $j \leq j'$.  If there is a walk of cost at most $B$ that covers $\opt$ colors, then there is a solution of cost at most $B$ that satisfies the packing constraints
\begin{equation*}
    y_b \le 1 \qquad \forall b \in [C]\ .
\end{equation*}

Using Theorem~\ref{thm:main} we can therefore find a solution to $I(s, \opt)$
satisfying $y_b \le \alpha = O\left(\frac{n^\epsilon}{\epsilon^2} \log C\right)$ with cost
at most $B$. This solution corresponds to a walk of cost at most $B$ that covers at least
$\opt / \alpha$ different colors. This finishes the proof of Theorem~\ref{thm:colorful-orienteering}. 

\section{Applications via augmentation}
\label{sec:app-aug}
While the direct applications for our framework, as in the previous two sections, require the solutions of the problem to have a recurrence
that can be embedded in a dynamic program like structure, it
is in fact applicable also to some problems that cannot directly
be written using such a recurrence.
\subsection{Robust Perfect Matching}

In this part we prove our result on the robust version of the perfect matching problem in bipartite graphs. We consider the graph $G=(U\cup V,E)$, where
we assume that $|U| = |V|$. 

First, we phrase the perfect matching problem as a max-flow problem in the standard way. For every vertex in the bipartite graph, we add a vertex in our flow network. We also add a source $s$ and a sink $t$. We connect $s$ to every vertex in $U$, we connect every vertex of $V$ to $t$. For every edge $(u,v)\in E$, we add an edge $(u,v)$ in the flow network. Every edge has capacity $1$. Clearly, there is a one to one mapping between flows of value $k$ in the network, and matchings of cardinality $k$ in the graph $G$.

We will augment the matching (starting from the empty matching) by a reduction to the Robust $s$-$t$ Path Problem. First, we need a few additional definitions. Let $f_t$ be the flow corresponding to the current matching at step $t$, and $f^*$ the maximum flow corresponding to the perfect matching that does not violate any of packing constraints. We denote by $G_t$ the residual network corresponding to matching $f_t$. This network is defined in the standard way, edge $e$ in the flow network has capacity $\capa_e-f_t(e)$ (i.e. $0$ if the edge is used, $1$ otherwise). For every edge $e=(u,v)$, we also add a reverse edge $e'=(v,u)$ of capacity equal to $f_t(e)$. The value of a feasible flow $f$ is defined as the total flow on the edges leaving the source $s$, which we denote by $v(f)$.

Assume that $v(f_t)< |U|$, then we can define a feasible flow $\Delta f_t$ of value $F := |U|-v(f_t)$ in the residual network $G_t$. For every edge $e=(u,v)$ used by both or neither of the flows $f_t$ and $f^*$, we set $\Delta f_t(e)=0$. For every other edge $e=(u,v)$ used by $f^*$ but not $f_t$, set $\Delta f_t(e)=1$. For every edge $e=(u,v)$ used by $f_t$ but not $f$, we set $\Delta f_t((v,u))=1$. Finally, we set $\Delta f_t(e)=0$ for all remaining edges in $G_t$. It is easy to verify that this is a feasible flow in $G_t$ of value $F$.

A standard result we will use is the fact that any feasible flow in such a flow network can be decomposed into a union of $\Delta$ edge-disjoint $s$-$t$ paths (see for instance Chapter 10.5 in \cite{erickson2023algorithms}). With this in mind, we can explain the augmentation procedure.

\paragraph{The augmentation procedure.} We find a set of $F$ many $s$-$t$ paths which will be approximately disjoint and do not violate the packing constraints by too much as follows. We create $F$ copies the graph $G_t$, with the sink $t_i$ of copy $i$ connecting to the source $s_{i+1}$ of the next copy.
Then we consider an instance $\mathcal P$ of Robust $s$-$t$ Path 
from $s_1$ to $t_F$ under the packing constraints below.
Recall that the solution vectors in Robust $s$-$t$ Path have one entry
$x\in \{0,1\}$ for every edge $e$, describing whether the edge is used.

First, for each edge $e$ in $G_t$, we introduce a constraint
\begin{equation*}
    \sum_{e'\in C(e)} x_{e'} \le 1\ ,
\end{equation*}
where $C(e)$ is the set of edges which are copies of the same edge $e\in G_t$. Also, for any packing constraint $\sum_{e\in M}a^j_e\le 1$ coming from the instance of the Matching Problem, we add a packing constraint 
\begin{equation*}
    \sum_{e\in E}\sum_{e'\in C(e)}x_{e'}\cdot a^j_e\le 1\ ,
\end{equation*}
where, by a slight abuse of notation, $C(e)$ is the set of all copies of $e$ in $G_t$, which corresponds to the edge $e$ in the original graph $G$. Note that edges corresponding to reverse edges of $G$ do not appear in any of these constraints. 

\begin{claim}
In instance $\mathcal P$, there exists a directed path from $s_1$ to $t_\Delta$ which does not violate any packing constraints.
\end{claim}
\begin{proof}
    We already established by standard flow arguments that there exists a set of $F$ disjoint $s$-$t$ paths in $G_t$ corresponding to the flow $\Delta f_t$. Therefore there clearly exists a path from $s_1$ to $t_F$ which does not violate any of the packing constraints of the first type. For the packing constraints coming from the matching instance itself, note that the only edges used by $\Delta f_t$ which appear in one of these constraints are edges used by $f^*$ but not $f_t$. Since $f^*$ does not violate any of the packing constraints, the flow $\Delta f_t$ also satisfies these constraints. Since an edge appears in the path decomposition of $\Delta f_t$ if and only if $\Delta f_t$ uses that edge, we can conclude the proof of the claim. 
\end{proof}

Using Theorem \ref{thm:main}, we can find in time $n^{O(1/\epsilon)}$ a path from $s_1$ to $t_F$ which does not violate any of the packing constraints by a factor more than $O\left(\frac{n^\epsilon}{\epsilon}\log (nk)\right)$. This path corresponds to a set of $F$ many $s$-$t$ path in $G_t$ which does not use any edge more than $O\left(\frac{n^\epsilon}{\epsilon}\log (nk)\right)$ times. Denoting by $P$ the support of this set of path, we know by scaling that there exists a feasible fractional flow of value $ \Delta/O\left(\frac{n^\epsilon}{\epsilon}\log (nk)\right)$ in $G_t$ using only edges of $P$. Therefore we compute a maximum flow in the graph $G_t$ restricted to edges that appear in $P$, and we augment $f_t$ using it. Because it uses only edges in the support of $P$, this ensures that no packing constraints of the Matching Problem increases by more than an additive $O\left(\frac{n^\epsilon}{\epsilon}\log (nk)\right))$ while doing this.

To measure how much the flow value is increased during each augmentation, we use the fact that any feasible fractional flow can be written as a convex combination of integral flows (see for instance Chapter 8 in \cite{schrijver2003course}). Therefore, in the residual graph restricted to $P$, there must exists an integral flow of value at least $F/O\left(\frac{n^\epsilon}{\epsilon}\log (nk)\right)$. Hence, by computing a maximum flow in that graph, we obtain a flow augmentation of value at least $F/O\left(\frac{n^\epsilon}{\epsilon}\log (nk)\right)$, so the missing flow is multiplied by a fraction at most 
\begin{equation*}
\left(1-1/\left(\frac{n^\epsilon}{\epsilon^2}  \cdot \log (nk)\right)\right)\ .
\end{equation*}
This proves that after $O\left(\frac{n^\epsilon}{\epsilon}\log^2 (nk)\right)$ augmentation steps, the missing flow value is $0$. At every step, the packing constraints are violated by an additive $O\left(\frac{n^\epsilon}{\epsilon}\log (nk)\right)$, therefore the total violation is at most $O\left(\frac{n^{2\epsilon}}{\epsilon^2}\cdot \log^3 (nk)\right)$. Rescaling $\epsilon$ and leveraging the one-to-one correspondence between maximum flows and maximum matchings concludes the proof of Theorem~\ref{thm:matching}.

\subsection{Santa Claus}
Let $\beta$ be the target approximation rate on the minimum value among players,
which we will specify later.
Using standard techniques of binary search, rounding and rescaling,
we can reduce the problem to the following variant.
\begin{itemize}
    \item Our goal is to give each player resources of total value at least $1/\beta$ or to determine that no solution exists that gives every player
a value of $1$,
    \item we assume that every value $v_{ij}$ is either $0, 1$ or $1/2^\ell$ with $1/\beta > 1/2^\ell > 1/2n$, and
    \item we restrict to solutions that assign to each player only resources of the same value.
\end{itemize}
Denote by $B$ the set of all possible values for resources. 
The reason why we do not need to consider resource values in $(1/\beta, 1)$
is that any such resource is sufficient for the player to obtain a value
of $1/\beta$; hence we can also set them to $1$. The reason we do not
need to consider value smaller than $1/2n$ is that we can round them
down to $0$, which will only remove a total value of at most $1/2$ from each player.

These reductions
come at a loss of a factor of $O(|B|) = O(\log n)$ in the approximation rate.
\paragraph*{General approach.}
We solve the problem by starting with the empty assignment
and then augmenting the solution in several iterations $k=1,2,3,\dotsc,\gamma$,
where $\gamma$ is the number of iterations that needs to be carefully bounded.
For each of these augmentations we will solve an instance of Integer Generalized Flow.
We say that a player is \emph{happy} in iteration $k$, if $i$ has a total value at least $1/(\alpha k) - 3k/\beta$.
Here, $\alpha$ is the approximation ratio for the Integer Generalized
Flow Problem.
Note that the threshold decreases over the iterations, which happens because already happy players
will lose some of their resources due to the later augmentations.
However, if by iteration $\gamma$, all players are happy
and we use $\beta = 6 \alpha\gamma^2$, then
we still retain a minimum value of
\begin{equation*}
	\frac{1}{\alpha\gamma} - \frac{3\gamma}{\beta} \ge \frac{1}{\alpha \gamma} - \frac{3\gamma}{6 \alpha\gamma^2} \ge \frac{1}{2 \alpha \gamma} \ge \frac{1}{\beta} \ .
\end{equation*}
In other words, if we can obtain a low number
of iterations $\gamma$,
then also our overall approximation ratio $\beta$ is low.
We will do the augmentation in such a way that each augmentation increases
the cost of the allocation by at most the cost of the optimal allocation
in expectation. Then, again, $\gamma$ gives a bound on the approximation
factor for the cost of the solution.

Throughout the algorithm we will maintain the invariants
that every player is only assigned resources of the same value
and an unhappy player is not assigned any resources.

\paragraph*{Augmentation.}
Assume we are given an assignment $\sigma: R \rightarrow P\cup\{\bot\}$, where $\bot$
indicates a resource is not assigned.
We assume $\sigma$ satisfies the previous invariants, but it does not make all players happy.
We define an instance of Integer Generalized Flow that represents augmentations.

The network contains the following vertices.
\begin{itemize}
    \item A vertex for each resource $j\in R$,
    \item a vertex for each player $i\in P$,
    \item a vertex $a(i, b)$ for every $i\in P$ and resource value $b\in B\setminus\{0,1\}$ that represents adding $1/b$ resources of value $b$ to $i$,
    \item a source vertex $s$, and
    \item a special vertex $\bot$ representing that a resource is not assigned.
\end{itemize}
Before defining the edges, we briefly summarize the intuition behind it:
each incoming unit of flow to a player $i$ means that in the augmentation to the optimal solution we 
move a total value of $1$ in resources to $i$. This may represent that $i$
is currently unhappy and through the augmentation $i$ becomes happy, or that the augmentation removes a resource $j$ of value $v_{ij} = 1$ from $i$ and we need to give $i$ other resources to compensate for the loss.
The total value of $1$ in resources that we give to $i$ can be by giving $i$ a new resource $j$ of value $v_{ij} = 1$, which we represent
by a unit flow from $i$ to $j$. It may also be by giving $i$ a total of $1/b$
many resources of value $b \in B\setminus\{0,1\}$.
We represent this by
sending a unit flow from $i$ to $a(i,b)$, where the edge has a gain of $1/b$.
The vertex $a(i, b)$ then has to propagate this flow to resources $j$ with $v_{ij} = b$.
For each unit of flow to a resource $j$ that is currently assigned to a player $i$,
we propagate the flow to $i$. With the intuition above this means we need 
to give $i$ new resources of value $1$ as a replacement. However, if $v_{ij} < 1$,
then this would not be accurate because in the augmentation to an optimal solution,
there might only be a value $v_{ij}$ moved to $i$.
Thus, for this case we set the gain of the edge to zero,
indicating that a small resource that is removed does not have to be replaced at all.
This is perhaps the most counter-intuitive aspect of this construction, but without
this, the structure of the flow would be too complicated to formally model it as an integer 
generalized flow. Fortunately, we will can cope with the inaccuracy resulting from this.

Formally, we introduce the following edges.
\begin{itemize}
    \item for each resource $j$ with $\sigma(j) \neq\bot$ and $v_{\sigma(j), j} = 1$: an edge $e = (j, \sigma(j))$ with $g(e) = 1$ and $w(e) = 0$,
    \item for each resource $j$ with $\sigma(j) = \bot$ or $v_{\sigma(j), j} < 1$, an edge $e = (j, \sigma(j))$ with $g(e) = 0$ and $w(e) = 0$,
    \item for each resource $j$ and player $i\neq \sigma(j)$ with $v_{ij} = 1$: an edge $e = (i, j)$
    with $g(e) = 1$ and $w(e) = c_{ij}$,
    \item for each resource $j$ and player $i\neq \sigma(j)$ with $v_{ij} = b \in B\setminus\{0,1\}$: an edge $e = (a(i, b), j)$ with $g(e) = 1$ and $w(e) = c_{ij}$,
    \item for each player $i$ and value $b \in B\setminus \{0,1\}$: an edge $e = (i, a(i, b))$ with $g(e) = 1/b$ and $w(e) = 0$, and
    \item for each $i\in P'$: an edge $(s, i)$ with $g(e) = 1$ and $w(e) = 0$.
\end{itemize}
The excess flow of $s$ is set to $|P'|$
and the capacity on each edge is $1$.
We will show in the analysis that
assuming there is an assignment that gives each player a value of $1$ of resources
of the same value, the instance of Integer Generalized Flow is feasible with a weight
no larger than the cost of that assignment, regardless
of $\sigma$.

We obtain a generalized flow $f$ with weight at most $B$ (the budget for the cost
in Santa Claus)
and where the flow on each edge is at most $\alpha$ using
Theorem~\ref{thm:gen-flow}. If this does not succeed we can safely return
that there is no optimal assignment of cost at most $B$ that gives each player a value of $1$ in resources of the same value.
For each $i\in P$ and $b \in B\setminus\{0,1\}$ let
\begin{equation*}
    y_{i,b} = \sum_{e\in \delta^-(a(i,b))} f(e) \ .
\end{equation*}
In order to use standard flow arguments, based on the fixed value of $y_{i,b}$ we construct
a classical flow network to replace the generalized flow.
Towards this, replace each vertex $a(i,b)$ by one source $a_s(i,b)$ and one sink $a_t(i,b)$,
where for incoming edges to $a(i,b)$ we replace $a(i,b)$ by $a_t(i,b)$ and for outgoing edges by $a_s(i,b)$.
Furthermore, we introduce a sink vertex $r(i,b)$ and each
edge $(j, i)$ with $i = \sigma(j) \neq \bot$ and $v_{ij} = b \in B\setminus\{0,1\}$ is replaced by
an edge $(j, r(i, b))$. Intuitively, flow to $r(i,b)$ corresponds to small resources \emph{removed}
from~$i$, whereas flow to $a_t(i,b)$ (and from $a_s(i,b)$) corresponds to \emph{adding} resources
of value $b$ to $i$.
By slight abuse of notation we consider $f$ to be a flow in the new
graph, where we apply to the flow the same modifications as to the edges.
Then $f$ is a solution to the following system, where $\delta^-(\cdot)$ and $\delta^+(\cdot)$
refer to the incoming and outgoing edges of a vertex in the modified network.
\begin{align*}
    \sum_{e\in \delta^-(j)} x_e &= \sum_{e\in \delta^+(j)} x_e &\forall j\in R \\
    \sum_{e\in \delta^-(i)} x_e &= \sum_{e\in \delta^+(i)} x_e &\forall i\in P \\
    \sum_{e\in\delta^+(s)} x_{e} &\ge |P'| & \\
    \sum_{e\in\delta^+(a_s(i,b))} x_{e} &\ge \frac{1}{b}\cdot y_{i,b} &\forall i\in P, b\in B\setminus\{0,1\} \\
    \sum_{e\in\delta^-(a_t(i,b))} x_{e} &\le y_{i,b} &\forall i\in P, b\in B\setminus\{0,1\} \\
    \sum_{e\in\delta^-(r(i,b))} x_{e} &\le \alpha \cdot |\delta^-(r(i,b))| &\forall i\in P, b\in B\setminus\{0,1\} \\
    x_{e} &\le \alpha &\forall e\in E \ .
\end{align*}
This mathematical system defines a classical flow problem with flow conservation on $R$ and $P$,
lower bounds on the sources $s, a_s(\cdot,\cdot)$
and upper bounds on the sinks $a_t(\cdot,\cdot), r(\cdot,\cdot)$. No restriction
on the flow to sink $\bot$ is given.

Decompose $f' + f'' = f$ where
$f'$ is a flow from $s$ to the sinks $a_t(\cdot,\cdot), r(\cdot,\cdot),\bot$ and $f''$ is a flow from the sources $a_s(\cdot,\cdot)$ to the sinks $a_t(\cdot,\cdot), r(\cdot,\cdot),\bot$. Then $f'/\alpha + f''/(k\alpha)$ is a fractional flow of weight still at most $B$ satisfying flow conservation on $P,R$ and the following bounds.
\begin{align*}
    \sum_{e\in\delta^+(s)} x_{e} &\ge \sum_{e\in\delta^+(s)}\frac{f'(e)}{\alpha} & \\
    \sum_{e\in\delta^+(a_s(i,b))} x_{e} &\ge \left\lfloor \frac 1 b \cdot \frac{1}{k\alpha} \cdot y_{i,b} \right\rfloor &\forall i\in P, b\in B\setminus\{0,1\} \\
    \sum_{e\in\delta^-(a_t(i,b))} x_{e} &\le y_{i,b}  &\forall i\in P, b\in B\setminus\{0,1\} \\
    \sum_{e\in\delta^-(r(i,b))} x_{e} &\le \left\lceil \frac{|\delta^-(r(i,b))|}{k} + \sum_{e\in \delta^-(r(i,b))} \frac{f'(e)}{\alpha} \right\rceil &\forall i\in P, b\in B\setminus\{0,1\} \\
    x_{e} &\le 1 &\forall e\in E \ .
\end{align*}
By standard flow arguments, we can in polynomial time randomly sample an integer flow $z$ satisfying
\begin{equation*}
    \E\left[\sum_{e\in\delta^+(s)} z_e\right] \ge \sum_{e\in\delta^+(s)}\frac{f'(e)}{\alpha}
    \quad \text{ and } \quad \E\left[\sum_{e\in E} w_e \cdot z_e \right] \le B \ ,
\end{equation*}
and all other integer bounds with probability $1$.
We augment $\sigma$ to $\sigma'$ using the flow $z$. More precisely,
for each resource $j$ with no outgoing and incoming flow, we set $\sigma'(j) = \sigma(j)$.
Consider now a resource $j$ 
with non-zero outgoing flow. Since $j$ has only one outgoing edge, the outgoing flow
and therefore also the incoming flow must be exactly $1$. We set $\sigma'(j) = i$
where $i$ is the player
such that the incoming flow either comes from vertex $i$ or from $a_s(i,b)$ for some $b\in B\setminus\{0,1\}$.
By the definition of the weights of the network, the increase in cost by the augmentation is at most the weight of $z$, which is at most $B$ in expectation.

We now repair the invariants. For every player that was happy in $\sigma$
and is no longer happy in $\sigma'$ we remove all resources.
If a player that became happy in $\sigma'$ and was unhappy in $\sigma$
receives resources of different values, we only keep resources of
a single value, choosing one arbitrarily.
\paragraph*{Feasibility of Generalized Flow instance.}
Suppose there exists an assignment $\sigma_{\OPT}$
that gives each player $i$ a total value
of at least $1$ from resources $j$ of the same value $v_{ij}$.
We will define a generalized flow $f$ for the instance above with weight
at most the cost of $\sigma_{\OPT}$.
For simplicity, we only guarantee
\begin{equation*}
    \sum_{e\in \delta^+(v)} f(e) - \sum_{e\in \delta^-(v)} g_e \cdot f(e) \ge \begin{cases}
        |P'| &\text{ if } v = s, \\
        0 &\text{ if } v \neq s \ .
    \end{cases}
\end{equation*}
In the definition of the Generalized Flow Problem equality is required.
However, one can easily derive another generalized flow that satisfies equality
by repeatedly reducing the flow on edges $e\in \delta^+(v)$ for vertices
$v$ where the inequality is strict. This maintains the previous inequalities,
only reduces weight, and must terminate, since we are only reducing flow.

Define $f$ as follows.
\begin{enumerate}
    \item For each player $i\in P'$ we set $f((s,i')) = 1$,\label{en:santa-unass}
    \item for each player $i\in P$ that in $\sigma_{\OPT}$ is assigned resources of value $b\in B\setminus\{0,1\}$ we set $f((i, a(i,b))) = 1$,\label{en:santa-small-out}
    \item for each resource $j\in R$ with $\sigma_{\OPT}(j) \neq \sigma(j)$ or $v_{\sigma(j),j} < 1$
    we set $f((j,\sigma(j)) = 1$,\label{en:santa-res-out}
    \item for each resource $j\in R$ where $\sigma_{\OPT}(j) \neq \sigma(j)$ and $v_{\sigma_{\OPT}(j), j} = 1$ 
    we set $f((\sigma_{\OPT}(j),j)) = 1$, \label{en:santa-big}
    \item for each resource $j\in R$ where $\sigma_{\OPT}(j) \neq \sigma(j)$ and $v_{\sigma_{\OPT}(j), j} = b\in B\setminus\{0,1\}$
    we set $f((a(\sigma_{\OPT}(j),b),j)) = 1$.\label{en:santa-small}
\end{enumerate}
All other edges have a flow of zero.
Clearly the source has excess of $|P'|$ as required.
Now consider a resource~$j$. If $j$ has incoming flow, then only one unit, because only one of the cases \ref{en:santa-big} or \ref{en:santa-small} can apply. In either case also \ref{en:santa-res-out} applies and therefore
the outgoing flow is also one.

Now consider a player $i \in P$.
We ignore incoming flow from resources $j$ with $v_{ij} < 1$, since the gain on those edges is zero.
If $i$ has incoming flow from other edges, then because of
\ref{en:santa-unass} or \ref{en:santa-res-out} and only one unit.
In either case, none
of the resources assigned to $i$ in $\sigma$ are assigned to $i$ in
$\sigma_{\OPT}$, which means that in $\sigma_{\OPT}$ either there is a resource $j\in R$ with $v_{ij} = 1$ and $i = \sigma_{\OPT}(j) \neq \sigma(j)$
and therefore by \ref{en:santa-big} an outgoing flow from $i$;
or in $\sigma_{\OPT}$ there are resources of some value $b\in B\setminus\{0,1\}$
assigned to $i$, in which case by \ref{en:santa-small-out} there is outgoing
flow from $i$.

Finally consider vertices $a(i,b)$ where $b\in B\setminus\{0,1\}$.
If $a(i,b)$ has incoming flow then because of \ref{en:santa-small-out}
and it must be exactly one unit from $i$, which has a gain of $1/b$.
This means there are resources of value $b$ assigned to $i$ in $\sigma_{\OPT}$
and since $\sigma_{\OPT}$ only assigns resources of the same value,
there must be $1/b$ many such resources assigned to $i$.
By \ref{en:santa-small} there is a total outgoing flow of $1/b$
from $a(i,b)$.

The only non-zero weight is on edges that incur in the cases~\ref{en:santa-big}
and~\ref{en:santa-small}. Here, the weight incurred in the flow is
the same as the cost incurred in $\sigma_{\OPT}$.

\paragraph*{Number of Iterations and approximation.}
By the augmentation some players become happy but also some players that were happy
may also become unhappy,
because we removed resources from them. The next two lemmas
allow us to lower bound the number of happy players after augmentation.
\begin{lemma}
    Suppose that player $i$ has incoming flow in $z$. Then $i$ is happy in $\sigma'$.
\end{lemma}
\begin{proof}
    Since $i$ has incoming flow, it must also have outgoing flow. If the outgoing flow is to a job $j'$,
    then $v_{ij'} = 1$ and $\sigma'(j') = i$, therefore $i$ is happy.
    If the outgoing flow is to $a_t(i,b)$ for some $b\in B\setminus\{0,1\}$ then
    $y_{i,b} \ge 1$. Therefore, the outgoing flow from $a_s(i,b)$ must be at
    least $\lfloor 1/b \cdot 1/(k\alpha) \rfloor$. This is the same as the number
    of resources $j'$ that have $v_{ij'} = b$ and are assigned to $i$ in $\sigma'$.
    The total value of these resources for $i$ is
    \begin{equation*}
        b \cdot \left\lfloor \frac 1 b \cdot \frac{1}{k\alpha} \right\rfloor \ge \frac{1}{k\alpha} - b \ge \frac{1}{k\alpha} - \frac{3k}{\beta} \ .
    \end{equation*}
    Hence, $i$ is happy.
    Note that the postprocessing removes resources, if $i$ gets
    resources of different values. In that case make the analysis above
    with respect to the resource value that $i$ retains.
\end{proof}
Thus, players that are assigned a resource of value $1$ in $\sigma$ remain happy in
$\sigma'$: if they lose
their resource, they have incoming flow and by the previous lemma remain happy.
Furthermore, players in $P'$ with flow from $s$ become happy.

However, some players may also change from happy to unhappy.
Towards this, we will analyze the players that are assigned resources of some value $b\in B\setminus \{0,1\}$ in $\sigma$.
\begin{lemma}
    Suppose that player $i$ is happy in $\sigma$, where it is assigned resources of value $b\in B\setminus\{0,1\}$. If $\sum_{e\in \delta^-(r(i,b))} f'(e) / \alpha \le 2$ then $i$ remains happy.
\end{lemma}
\begin{proof}
    The number of resources that are assigned to $i$ in $\sigma$ is
    \begin{equation*}
        |\delta^-(r(i,b))| \ge \frac{1}{b} \left(\frac{1}{\alpha (k-1)} - \frac{3(k-1)}{\beta} \right) \ .
    \end{equation*}
    Out of these, $i$ retains in $\sigma'$ at least
    \begin{align*}
    |\delta^-(r(i,b))| - \sum_{e\in\delta^-(r(i,b))} z_e &\ge 
    |\delta^-(r(i,b))| - \left\lceil \frac{|\delta^-(r(i,b))|}{k} + \sum_{e\in \delta^-(r(i,b))} \frac{f'(e)}{\alpha} \right\rceil \\
    &\ge 
        |\delta^-(r(i,b))| \left(1 - \frac{1}{k}\right) - 3 \\
        &\ge \frac{1}{b} \left( \left(1 - \frac{1}{k}\right)\left(\frac{1}{\alpha (k-1)} - \frac{3(k-1)}{\beta}\right) - \frac{3}{\beta} \right) \\
        &\ge \frac{1}{b} \left( \left(1 - \frac{1}{k}\right)\left(\frac{1}{\alpha (k-1)}\right) - \frac{3k}{\beta}\right) \\
        &\ge \frac{1}{b} \left(\frac{1}{\alpha k} - \frac{3k}{\beta} \right) \ .
    \end{align*}
    Since each of these resources has value $b$ for $i$, $i$ remains happy.
\end{proof}
Thus, the number of happy players in $\sigma$ that become unhappy in $\sigma'$
is at most
\begin{equation*}
    \frac{1}{2} \sum_{i\in P} \sum_{b\in B\setminus\{0,1\}} \sum_{e\in \delta^-(r(i,b))}  \frac{f'(e)}{\alpha} \le \frac{1}{2} \sum_{e\in\delta^+(s)} \frac{f'(e)}{\alpha} \ .
\end{equation*}
Recall that the flow of $z$ leaving $s$ lower bounds
the number of players that were unhappy and become happy. Therefore, in expectation the number of happy players increases by at least
\begin{equation*}
\E\left[\sum_{e\in\delta^+(s)} z_e\right] - \frac{1}{2} \sum_{e\in\delta^+(s)} \frac{f'(e)}{\alpha} \ge \sum_{e\in\delta^+(s)} \frac{f'(e)}{\alpha} - \frac{1}{2} \sum_{e\in\delta^+(s)} \frac{f'(e)}{\alpha} \ge \frac{1}{2}\sum_{e\in\delta^+(s)} \frac{f'(e)}{\alpha} \ge \frac{1}{2 \alpha} |P'| \ .
\end{equation*}
Thus, the expected number of unhappy players in iteration $k$
is at most $(1 - 1/2\alpha)^k \cdot n$.
It follows that for some $\gamma = O(\alpha \log n)$
the expected number of unhappy players by iteration $\gamma$
is at most $1/100$. By Markov's inequality, the probability that
all players are happy by iteration $\gamma$ is at least $1 - 1/100$.
Setting $\alpha = O(n^\epsilon \cdot \log^3(n))$ and
using our previous bound of $\beta = O(\alpha\gamma^2)$ and the additional
loss of a factor of $O(|B|)$, we obtain the approximation rate
\begin{equation*}
    O(\beta |B|) \le O(\alpha\gamma^2 \log(n))\le O(\alpha^3 \log^2(n))\le O(n^{3\epsilon} \log^{11}(n)) \ .
\end{equation*}
By rescaling $\epsilon$ we achieve the tradeoff as in Theorem~\ref{thm:santa}.

\section*{Acknowledgments}
The authors would like to thank Bundit Laekhanukit, Nick Fischer, and Karthik C.\ S.\ for insightful discussions.

\bibliographystyle{alpha}
\bibliography{biblio}

\appendix
\section{Integrality Gap for Robust Matching}
\label{sec:integrality-gap}
One may hope that a simpler approach based on the naive LP relaxation
already gives a good approximation for Robust Bipartite Maximum
Matching. In this section, we will show that the integrality gap
of the naive relaxation can be very high.

For the integrality gap example, we focus on perfect matchings and consider the LP
\begin{align*}
    \sum_{e\in E} \ell^j(e) \cdot x_e &\le L_j &\forall j\in\{1,2,\dotsc,k\} \\
    \sum_{e\in \delta^-(v)} x_e &= 1 &\forall v\in A\cup B \\
    x_e &\ge 0 &\forall e\in E \ .
\end{align*}
Integral solutions for this linear program form perfect matchings 
that satisfy the additional length constraints.
Let $k\in\N$ and
consider a graph that contains $k$ many disjoint
paths $v^{(j)}_1,v^{(j)}_2,\dotsc,v^{(j)}_{2k+1}$, $j\in \{1,2,\dotsc,k\}$.
There are edges from $v^{(j)}_i$ to $v^{(j)}_{i+1}$ for each $j\in\{1,2,\dotsc,k\}$
and $i\in\{1,2,\dotsc,2k\}$.
Furthermore, there is one vertex $s$ and $k-1$ vertices $t_1,\dotsc,t_{k-1}$.
There is an edge from $s$ to $v^{(j)}_1$ and one from $v^{(j)}_{2k+1}$ to $t_i$
for each $j\in \{1,2,\dotsc,k\}$ and $i\in\{1,2,\dotsc,k-1\}$.
The graph is bipartite, with bipartition $A = \{s\} \cup \{v^{(j)}_{2i} : i\in \{1,2,\dotsc,k\}\}$ and $B = \{t_1,\dotsc,t_{k-1}\} \cup \{v^{(j)}_{2i+1} : i\in \{0,2,\dotsc,k\}\}$.

There are $k$ length functions and the $j$th length function $\ell^j$
has $\ell^j((v^{(j)}_{2i}, v^{(j)}_{2i+1})) = 1$ for each $i\in \{1,2,\dotsc,k\}$
and zero length on all other edges.
The budget of each length function is $L_j = 1$.
Each perfect matching must match $s$ to $v^{(j)}_1$ for some $j\in\{1,2,\dotsc,k\}$.
Thus, it must also match $v^{(j)}_{2i}$ to $v^{(j)}_{2i+1}$ for each $i\in\{1,2,\dotsc,k\}$, which incurs a length of $k$ in the $j$th length function.
In all other paths $j'\neq j$ we must match $v^{(j')}_{2i-1}$ to $v^{(j')}_{2i}$
for each $i\in\{1,2,\dotsc,k\}$ and $v^{(j')}_{2k+1}$ to some $t_{\ell}$ (it does not matter which one).
Thus, there is no perfect matching with a violation less than $k$.
At the same time the LP above is feasible:
For each $j\in\{1,2,\dotsc,k\}$, we set
\begin{itemize}
    \item $x_{(s,v^{(j)}_1)} = 1/k$,
    \item $x_{(v^{(j)}_{2i}, v^{(j)}_{2i+1})} = 1/k$ for each $i\in\{1,2,\dotsc,k\}$,
    \item $x_{(v^{(j)}_{2i-1}, v^{(j)}_{2i})} = 1 - 1/k$ for each $i\in\{1,2,\dotsc,k\}$,
and 
    \item $x_{(v^{(j)}_{2k+1}, t_{\ell})} = (1 - 1/k)/(k-1)$ for each $\ell\in\{1,2,\dotsc,k-1\}$.
\end{itemize}
It is easy to check that this solution is feasible.
Note that the number of vertices is $O(k^2)$. Thus, we obtain the following.
\begin{lemma}
    The integrality gap of the naive linear programming relaxation
    for Robust Bipartite Maximum Matching is at least $\Omega(\sqrt{n})$.
\end{lemma}
\end{document}